\documentclass[graybox]{svmult}

\usepackage{mathptmx}       % selects Times Roman as basic font
\usepackage{helvet}         % selects Helvetica as sans-serif font\citet
\usepackage{courier}        % selects Courier as typewriter font
\usepackage{type1cm}        % activate if the above 3 fonts are
\usepackage{makeidx}         % allows index generation
\usepackage{graphicx}        % standard LaTeX graphics tool
\usepackage{multicol}        % used for the two-column index
\usepackage[bottom]{footmisc}% places footnotes at page bottom

\usepackage{amsmath}
\usepackage{amssymb}
\usepackage{enumerate}
\usepackage[round]{natbib}
\usepackage{color}

\newcommand{\F}{\mathcal{F}}

\def\R{{\mathbb{R}}}

\newtheorem{lem}{Lemma}[section]
\newtheorem{prop}{Proposition}[section]
\newtheorem{df}{Definition}[section]
\newtheorem{cor}{Corollary}[section]
\newtheorem{rem}{Remark}[section]

\newcommand{\beq}{\begin{equation}}
\newcommand{\eeq}{\end{equation}}
\newcommand{\bi}{\begin{itemize}}
\newcommand{\bd}{\begin{description}}
\newcommand{\ei}{\end{itemize}}
\newcommand{\ed}{\end{description}}

\newcommand{\bc}{\begin{center}}
\newcommand{\ec}{\end{center}}

\newcommand{\sg}{\sigma}
\newcommand{\al}{\alpha}

\newcommand{\gm}{\gamma}

\newcommand{\ld}{\lambda}

\newcommand{\Ind}[1]{{\bf{1}}_{\{{#1}\}}}

\makeindex             % used for the subject index
\begin{document}

\title*{Derivative pricing for a multi-curve extension of the Gaussian, exponentially quadratic short rate model}
\titlerunning{Derivative pricing for a multi-curve Gaussian, exponentially quadratic short rate model}
\author{Zorana Grbac and Laura Meneghello and Wolfgang J. Runggaldier}
\institute{Zorana Grbac \at Laboratoire de Probabilit\'es et Mod\`eles Al\'eatoires, Universit\'e Paris Diderot - Paris 7, France \\ \email{grbac@math.univ-paris-diderot.fr}
\and Laura Meneghello  \at Dipartimento di Matematica Pura ed Applicata, Universit\`a di Padova,  Via Trieste 63,  I-35121-Padova \\ Present affiliation: Gruppo Banco Popolare  \textit{Disclaimer: The views, thoughts and opinions expressed in this paper are those of the authors in their individual
capacity and should not be attributed to Gruppo Banco Popolare or to the authors as representatives or employees
of Gruppo Banco Popolare}   \\ \email{meneghello.laura@yahoo.com}
\and Wolfgang J. Runggaldier \at Dipartimento di Matematica Pura ed Applicata, Universit\`a di Padova,  Via Trieste 63,  I-35121-Padova \\ \email{runggal@math.unipd.it}
}
\maketitle

\abstract{The recent financial crisis has led to so-called
multi-curve models for the term structure.  Here we study a
multi-curve extension of short rate models where, in addition to the
short rate itself, we introduce short rate spreads. In particular,
we consider a Gaussian factor model where the short rate and the
spreads are second order polynomials of Gaussian factor processes.
This leads to an exponentially quadratic model class that is less
well known than the exponentially affine class. In the latter class the factors
enter linearly and for positivity one considers square root
factor processes.  While the square root factors in the affine class
have more involved distributions, in the quadratic class the factors
remain Gaussian and this leads to various advantages, in particular
for derivative pricing. After some preliminaries on martingale
modeling in the multi-curve setup, we concentrate on pricing of
linear and optional derivatives. For linear derivatives, we exhibit
an adjustment factor that allows one to pass from pre-crisis single
curve values to the corresponding post-crisis multi-curve values. }

%\noindent{\small\bf Mathematics Subject Classification :} Primary  91G30; Secondary  91G20, 34B30.

\noindent{\small\bf Keywords :} multi-curve models, short rate models, short rate spreads, Gaussian exponentially quadratic models, pricing of linear and optional
interest rate derivatives, Riccati equations, adjustment factors.

\section{Introduction}\label{S.0}

The recent financial crisis has heavily impacted the financial
market and the fixed income markets in particular. Key features put
forward by the crisis are counterparty and liquidity/funding risk.
In interest rate derivatives the underlying rates are typically
Libor/Euribor. These are determined by a panel of banks and thus
reflect various risks in the interbank market, in particular
counterparty and liquidity risk. The standard no-arbitrage relations
between Libor rates of different maturities have broken down and
significant spreads have been observed between Libor rates of
different tenors, as well as between Libor and OIS swap rates, where
OIS stands for Overnight Indexed Swap. For more details on this issue see equations \eqref{4}-\eqref{5a} and the paragraph following them, as well as the paper by \citet{BormettiBrigoFrancischelloPallavicini15} and a corresponding version in this volume. This has led practitioners
and academics alike to construct multi-curve models where future cash
flows are generated through curves associated to the underlying
rates (typically the Libor, one for each tenor structure), but are
discounted by another curve.

For the pre-crisis single-curve setup various interest rate models have
been proposed. Some of the standard model classes are: the short
rate models; the instantaneous forward rate models in an
Heath-Jarrow-Morton (HJM) setup; the market forward rate models
(Libor market models). In this paper we consider a possible
multi-curve extension of the short rate model class that, with
respect to the other model classes, has in particular the advantage
of leading more easily to a Markovian structure. Other multi-curve
extensions of short rate models  have appeared in the literature
such as \citet{KijimaTanakaWong09}, \citet{Kenyon10},
\citet{FilipovicTrolle13} and  \citet{MorinoRunggaldier14}. The present paper considers an exponentially quadratic model, whereas the models in the mentioned papers concern mainly the exponentially affine framework, except for \citet{KijimaTanakaWong09} in which the exponentially quadratic models are mentioned. More details on the difference between the exponentially affine and exponentially quadratic short rate models will be provided below.

%We mention that other possible approaches to extend the classical one-curve models
%to a multi curve setting are: i) modeling directly the Libor rates
%in the spirit of the LMM; ii) HJM modeling (see [book]). Here we
%shall be concerned with the extension of the short rate modeling as
%introduced above and that has previously been considered in e.g.
%\citet{Kenyon10}, \citet{KijimaTanakaWong09}, \citet{FilipovicTrolle13}). One of
%the major advantages of the short rate approach is that it leads
%more easily to a Markovian structure.
%

Inspired by a credit risk analogy, but also by a common practice of
deriving multi-curve quantities by adding a spread over the
corresponding single-curve risk-free quantities, we shall consider,
next to the short rate itself, a short rate spread to be added to
the short rate, one for each possible tenor structure. Notice that
these spreads are added from the outset.

To discuss the basic ideas in an as simple as possible way, we
consider just a two-curve model, namely with one curve for
discounting and one for generating future cash flows; in other
words, we shall consider a single tenor structure. We shall thus
concentrate on the short rate $r_t$ and a single short rate spread
$s_t$ and, for their dynamics, introduce a factor model. In the
pre-crisis single-curve setting there are two basic factor model
classes for the short rate: the exponentially affine and the
exponentially quadratic model classes. Here we shall concentrate on
the less common quadratic class with Gaussian factors. In the
exponentially affine class where, to guarantee positivity of rates
and spreads, one considers generally square root models for the
factors, the distribution of the factors is $\chi^2$. In the 
exponentially quadratic class the factors have a more convenient
Gaussian distribution.

The paper is structured as follows. In the preliminary section
\ref{S.1} we mainly discuss issues related to martingale modeling.
In section \ref{S.2} we introduce the multi-curve  Gaussian,
exponentially quadratic model class. In section  \ref{S.3} we deal
with pricing of linear interest rate derivatives and, finally, in
section \ref {S.4} with nonlinear/optional interest rate
derivatives.

%To discuss the basic ideas in a simple
%way, here we consider just a two-curve model, namely with a curve
%for discounting and one single Libor, for a given tenor, to generate future cash flows. %without collateralization.

\section{Preliminaries}\label{S.1}

\subsection{Discount curve and collateralization.}\label{S.1.1}

In the presence of multiple curves, the choice of the curve for
discounting the future cash flows, and a related choice of the
numeraire for the standard martingale measure used for pricing, in
other words, the question of absence of arbitrage, becomes
non-trivial (see e.g. the discussion in \citet{KijimaMuromachi15}). To avoid issues of arbitrage, one should possibly have
a common discount curve to be applied to all future cash flows
independently of the tenor. A choice, which has been widely accepted
and became practically standard, is given by the OIS-curve $T
\mapsto p(t,T)=p^{OIS}(t,T)$ that can be stripped from OIS rates,
namely the fair rates in an OIS. The arguments justifying this
choice and which are typically evoked in practice, are the fact that
the majority of the traded interest rate derivatives are nowadays
being collateralized and the rate used for remuneration of the
collateral is exactly the overnight rate, which is the rate the OIS
are based on. Moreover, the overnight rate bears very little risk
due to its short maturity  and therefore can be considered
relatively risk-free. In this  context we also point out that
prices, corresponding to fully collateralized transactions, are
considered as clean prices (this terminology was first introduced by
\citet{Crepey11b2} and \citet{BieleckiCrepeyBrigo14}). Since
collateralization is by now applied in the majority of cases, one
may thus ignore counterparty and liquidity risk between individual
parties when pricing interest rate derivatives, but cannot ignore the counterparty and liquidity risk in the interbank market as a whole. These risks are often jointly referred to as interbank risk and they are main drivers of the multi-curve phenomenon, as documented in the literature (see e.g. \citet{CrepeyDouady12}, \citet{FilipovicTrolle13} and \citet{GallitschkeMuellerSeifried14}). 
%the multi-curve phenomenon. The multi-curve phenomenon itself is a consequence of counterparty and liquidity risk in the interbank market as documented in the literature, see e.g. \citet{CrepeyDouady12}, \citet{FilipovicTrolle13} and \citet{GallitschkeMuellerSeifried14}. 
We shall thus consider only {\sl clean
valuation} formulas, which take into account the multi-curve issue. %namely without counterparty risk. 
 Possible ways to account for counterparty risk and funding issues between individual counterparties in a contract are, among
others, to follow a global valuation approach  that leads to nonlinear derivative
valuation (see
\citet{BrigoPallaviciniPerini12}, \citet{BrigoMoriniPallavicini12}
and other references therein, and in particular \citet{PallaviciniBrigo13} for a global valuation approach applied specifically to interest rate modeling), or to consider various valuation adjustments that are
generally computed on top of the clean prices (see
\citet{Crepey11b2}). A fully nonlinear valuation is preferable, but
is more difficult to achieve. On the other hand, valuation
adjustments are more consolidated and also used in practice and this
gives a further justification to still look for clean prices. Concerning the explicit role of collateral in the pricing of interest rate derivatives, we refer to the above-mentioned paper by  \citet{PallaviciniBrigo13}.

\subsection{Martingale measures}\label{S.1.2}

The fundamental theorem of asset pricing links the economic
principle of absence of arbitrage with the notion of a martingale
measure. As it is well known, this is a measure, under which the
traded asset prices, expressed in units of a same numeraire, are
local martingales. Models for interest rate markets are typically
incomplete so that absence of arbitrage admits many martingale
measures. A common approach in interest rate modeling is to perform
martingale modeling, namely to model the quantities of interest
directly under a generic martingale measure; one has then to perform
a calibration in order to single out the specific martingale measure
of interest. The modeling under a martingale measure now imposes
some conditions on the model and, in interest rate theory, a typical
such condition is the Heath-Jarrow-Morton (HJM) drift condition.

Starting from the OIS bonds, we shall first derive a suitable
numeraire and then consider as martingale measure a measure $Q$
under which not only the OIS bonds, but also the FRA contracts seen
as basic quantities in the bond market, are local martingales when
expressed in units of the given numeraire. To this basic market one
can then add various derivatives imposing that their prices,
expressed in units of the numeraire, are local martingales under
$Q$.

%In view of arbitrage-free modeling, we have to start from the basic
%quantities present in the bond market, namely the OIS bonds and the
%FRA contracts. We thus need to define our models under a measure, for which
%discounted OIS bonds and the FRA contracts are local martingales. We
%shall in fact first consider a model for the OIS bond prices, which
%is set up directly under a martingale measure and then add the FRA
%contracts thereby ensuring that the model remains arbitrage free. At
%this point one can then add also various derivatives, always
%ensuring that the global model remains arbitrage free.

Having made the choice of the OIS curve $T \mapsto p(t, T)$ as the
discount curve, consider the instantaneous forward rates  $f(t, T):=
-\frac{\partial}{\partial T} \log p(t, T)$ and let
 $r_t=f(t, t)$ be the corresponding short rate at the generic time
$t$. Define the OIS bank account as \beq\label{1} B_t = \exp \left(
\int_0^t r_s ds \right)\eeq and, as usual, the standard martingale
measure $Q$ as the measure, equivalent to the physical measure $P$,
that is associated to the bank account $B_t$ as numeraire. Hence the
arbitrage-free prices of all assets, discounted by $B_t$, have
to be local martingales with respect to $Q$. For derivative pricing,
among them also FRA pricing, it is often more convenient to use,
equivalently, the forward measure $Q^T$ associated to the OIS bond
$p(t,T)$ as numeraire. The two measures $Q$ and $Q^T$ are related by
their Radon-Nikodym density process \beq\label{2}\frac{d\,
Q^{T}}{d\, Q} \Big|_{\mathcal{F}_t}=\frac{p(t, T)}{B_t  p(0, T) }
\qquad 0 \leq t \leq T.\eeq As already mentioned, we shall follow the
traditional {\sl martingale modeling}, whereby the model dynamics
are assigned under the martingale measure $Q$. This leads to
defining the OIS bond prices according to
\begin{equation}\label{3}p(t,T)=E^Q\left\{\exp\left[-\int_t^Tr_udu\right]\mid\F_t\right\}\end{equation}
after having specified the $Q-$dynamics of $r$.

Coming now to the FRA contracts, recall that they concern a forward
rate agreement, established at a time $t$ for a future interval
$[T,T+\Delta]$, where at time $T+\Delta$ the interest corresponding
to a floating rate is received in exchange for the interest corresponding to
a fixed rate $R$. There exist various possible conventions
concerning the timing of the payments. Here we choose payment in
arrears, which in this case means at time $T+\Delta$. Typically, the
floating rate is given by the Libor rate and, having assumed
payments in arrears, we also assume that the rate is fixed at the
beginning of the interval of interest, here at $T$. Recall that for
expository simplicity we had reduced ourselves to a two-curve setup
involving just a single Libor for a given tenor $\Delta$.
%We denote by $L(t;T,T+\Delta)$ the single Libor for the time interval $[T,T+\Delta]$, evaluated at $t\le T$.
The floating rate received at $T+\Delta$ is therefore the
rate $L(T;T,T+\Delta)$, fixed at the inception time $T$. For a unitary notional, and using the $(T+\Delta)-$forward
measure $Q^{T+\Delta}$ as the pricing measure, the arbitrage-free
price at $t\leq T$ of the FRA contract is then
\beq
\label{eq:FRA-price}
P^{FRA}(t;T,T+\Delta,R)=\Delta
p(t,T+\Delta)E^{T+\Delta}\left\{L(T;T,T+\Delta)-R\mid\F_t\right\},
\eeq
where $E^{T+\Delta}$ denotes the expectation with respect to the measure $Q^{T+\Delta}$.
From this expression it follows that the value of the fixed rate $R$
that makes the contract fair at time $t$ is given by
\beq\label{4}R_t=E^{T+\Delta}\left\{L(T;T,T+\Delta)\mid\F_t\right\}:=
L(t;T,T+\Delta)\eeq and we shall call $L(t;T,T+\Delta)$ the {\sl
forward Libor rate}. Note that $L(\cdot;T,T+\Delta)$  is a $Q^{T+\Delta}-$martingale by
construction.

In view of developing a model for $L(T;T,T+\Delta)$, recall that, by
absence of arbitrage arguments, the classical discrete compounding
forward rate at time $t$ for the future time interval $[T,T+\Delta]$
is given by
$$F(t;T,T+\Delta) =
\frac{1}{\Delta}\,\left(\frac{p(t,T)}{p(t,T+\Delta)}-1\right),$$
where $p(t,T)$ represents here the price of a risk-free zero coupon
bond. This expression can be justified also by the fact that it
represents the fair fixed rate in a forward rate agreement, where
the floating rate received at $T+\Delta$ is
\beq\label{5}F(T;T,T+\Delta) =
\frac{1}{\Delta}\,\left(\frac{1}{p(T,T+\Delta)}-1\right)\eeq and we
have \beq\label{5a}
F(t;T,T+\Delta)=E^{T+\Delta}\left\{F(T;T,T+\Delta)\mid\F_t\right\}.\eeq
This makes the forward rate coherent with the risk-free bond prices,
where the latter represent the expectation of the market concerning
the future value of money.

Before the financial crisis, $L(T;T,T+\Delta)$ was assumed to be
equal to $F(T;T,T+\Delta)$, an assumption that allowed for various
simplifications in the determination of derivative prices.  After
the crisis $L(T;T,T+\Delta)$ is no longer equal to $F(T;T,T+\Delta)$
and what one considers for $F(T;T,T+\Delta)$ is
in fact the {\sl OIS discretely compounded rate}, which is based on the OIS bonds, even though the
OIS bonds are not necessarily equal to the risk-free bonds (see sections 1.3.1 and 1.3.2 of \citet{GrbacRunggaldier15} for more details on this issue). In
particular, the Libor rate $L(T;T,T+\Delta)$ cannot be expressed by
the right hand side of (\ref{5}). The fact that
$L(T;T,T+\Delta)\not=F(T;T,T+\Delta)$ implies by (\ref{4}) and
(\ref{5a}) that also $L(t;T,T+\Delta)\not=F(t;T,T+\Delta)$ for all
$t\le T$ and this leads to a {\sl Libor-OIS spread}
$L(t;T,T+\Delta)-F(t;T,T+\Delta)$.

Following some of the recent literature (see e.g.
\citet{KijimaTanakaWong09}, \citet{CrepeyGrbacNguyen12},
\citet{FilipovicTrolle13}), one possibility is now to keep the
classical relationship (\ref{5}) also for $L(T;T,T+\Delta)$ thereby
replacing however the bonds $p(t,T)$ by fictitious risky ones $\bar
p(t,T)$ that are assumed to be affected by the same factors as the
Libor rates. Such a bond  can be seen as an average bond issued by a
representative bank from the Libor group and it is therefore
sometimes referred to in the literature as a {\sl Libor bond}. This
leads to
\beq\label{6}L(T;T,T+\Delta)=\displaystyle\frac{1}{\Delta}\,\left(\frac{1}{\bar
p(T,T+\Delta)}-1\right).\eeq Recall that, for simplicity of
exposition, we consider a single Libor for a single tenor $\Delta$
and so also a single fictitious bond. In general, one has one Libor
and one fictitious bond for each tenor, i.e.
$L^{\Delta}(T;T,T+\Delta)$ and $\bar p^{\Delta}(T,T+\Delta)$.  Note
that we shall model the bond prices $\bar p(t, T)$, for all $t$ and
$T$ with $t\leq T$, even though only the prices $\bar
p(T,T+\Delta)$, for all $T$, are needed in relation (\ref{6}).
Moreover, keeping in mind that the bonds $\bar p(t,T)$ are
fictitious, they do not have to satisfy the boundary condition $\bar
p(T,T)=1$, but we still assume this  condition in order to simplify
the modeling.

To derive a dynamic model for $L(t;T,T+\Delta)$, we may now derive a
dynamic model for $\bar p(t,T+\Delta)$, where we have to keep in
mind that the latter is not a traded quantity. Inspired by a
credit-risk analogy, but also by a common practice of deriving
multi-curve quantities by adding a spread over the corresponding
single-curve (risk-free) quantities, which in this case is the short
rate $r_t$, let us define then the Libor (risky) bond prices as
\beq\label{71}\bar
p(t,T)=E^Q\left\{\exp\left[-\int_t^T(r_u+s_u)du\right]\mid\F_t\right\},\eeq
with $s_t$ representing the short rate spread. In case of default
risk alone, $s_t$ corresponds to the hazard rate/default intensity,
but here it corresponds more generally to all the factors affecting
the Libor rate, namely besides credit risk, also liquidity risk etc.
Notice also that the spread is introduced here from the outset.
Having for simplicity considered a single tenor $\Delta$ and thus a
single $\bar p(t,T)$, we shall also consider only a single spread
$s_t$. In general, however, one has a spread $s_t^{\Delta}$ for each
tenor $\Delta$.

We need now a dynamical model for both $r_t$ and $s_t$ and we shall
define this model directly under the martingale measure $Q$ ({\sl
martingale modeling}).
%With the OIS bond prices as defined above in
%(\ref{3}) and the forward Libor rates as given in (\ref{4}), we obtain an
%arbitrage-free system for the basic quantities, namely OIS bonds and
%FRA contracts. Notice that, since the FRA prices are given in terms of the
%fair FRA rates $R_t$, see (\ref{eq:FRA-price}) and (\ref{4}), %by the
%values of the fixed rate $R$ that makes the given FRA fair,
%their martingale property follows by construction.
%To define then
%arbitrage-free prices of the various derivatives it thus suffices to
%price them under various forward measures.

\section{Short rate model}\label{S.2}

\subsection{The model}\label{S.2.1}
As mentioned, we shall consider a dynamical model for $r_t$ and the
single spread $s_t$ under the martingale measure $Q$ that, in
practice, has to be calibrated to the market. For this purpose we
shall consider a factor model with several factors driving $r_t$ and
$s_t$.

The two basic factor model classes for the short rate in the pre-crisis single-curve setup, namely the exponentially affine
and the exponentially quadratic model classes, both allow
for flexibility and analytical tractability and this in turn allows
for closed or semi-closed formulas for linear and optional interest
rate derivatives. The former class is usually better known than the
latter, but the latter has its own advantages. In fact, for the
exponentially affine class one would consider $r_t$ and $s_t$ as
given by a linear combination of the factors and so, in order to
obtain positivity, one has to consider a square root model for the
factors. On the other hand, in the Gaussian exponentially quadratic
class, one considers mean reverting Gaussian factor models, but at
least some of the factors in the linear combination for $r_t$ and
$s_t$ appear as a square. In this way the distribution of the
factors remains always Gaussian; in a square-root model it is a
non-central $\chi^2-$distribution. Notice also that the
exponentially quadratic models can be seen as dual to the square
root exponentially affine models.

In the pre-crisis single-curve setting, the exponentially quadratic 
models have been considered  e.g. in  
\citet{ElKarouiMyneniViswanathan92}, \citet{Pelsser97},
\citet{GombaniRunggaldier01}, \citet{LeippoldWu02}, \citet{ChenFilipovicPoor04},
and \citet{Gaspar04}. However, since the pre-crisis exponentially affine models are
more common, there have also been more attempts to extend them to a
post-crisis multi-curve setting (for an overview and details see e.g.
\citet{GrbacRunggaldier15}). A first extension of exponentially
quadratic models to a multi-curve setting can be found in
\citet{KijimaTanakaWong09} and the present paper is devoted to a possibly
full extension.

Let us now present the model for $r_t$ and $s_t$, where
we consider not only the short rate $r_t$ itself, but also its
spread $s_t$ to be given by a linear combination of the factors,
where at least some of the factors appear as a square. To keep the
presentation simple, we shall consider a small number of factors
and, in order to model also a possible correlation between $r_t$ and
$s_t$, the minimal number of factors is three. It also follows from
some of the econometric literature that a small number of factors
may suffice to adequately model most situations (see also \citet{Duffee99} and \citet{DuffieGarleanu01}).
%[Duffee, G. (1999), Estimating the price of default risk, Review of Financial
%Studies 12, 197-226] and [Duffie, D. and N.G\^arleanu (2001), Risk
%and valuation of collateralized debt obligations, Financial Analysts
%Journal 57, 41-59].)

Given three independent affine factor processes
$\Psi^i_t,\>i=1,2,3$, having under $Q$ the Gaussian dynamics
\beq\label{7} d\Psi_t^i=-b^i\Psi_t^idt+\sg^i\,dw_t^i,\quad
i=1,2,3,\eeq with $b_i, \sigma_i > 0$ and  $w_t^i$, $i=1, 2, 3$,
independent $Q-$Wiener processes, we let
\beq\label{8}\left\{\begin{array}{lcl}
r_t&=&\Psi^1_t+(\Psi^2_t)^2\\
s_t&=&\kappa\Psi_t^1+(\Psi_t^3)^2\end{array}\right.,\eeq where
$\Psi_t^1$ is the common systematic factor allowing for
instantaneous correlation between $r_t$ and $s_t$ with correlation
intensity $\kappa$ and  $\Psi_t^2$ and $\Psi_t^3$ are the
idiosyncratic factors. Other factors may be added to drive $s_t$,
but the minimal model containing common and idiosyncratic components
requires three factors, as explained above. The common factor is particularly important
because we want to take into account the realistic feature of
non-zero correlation between $r_t$ and $s_t$ in the model.
%   In presenting joint models for $r_t$ and
%$s_t$ we also want to allow for non-zero correlation between $r_t$
%and $s_t$. It is obtained by considering common systematic factors,
%the remaining ones being idiosyncratic factors. The minimal model to
%achieve this requires three factors.

\begin{rem}\label{R.1} The zero mean-reversion level is here
considered only for convenience of simpler formulas, but can be
easily taken to be positive, so that short rates and spreads can
become negative only with small probability (see \citet{KijimaMuromachi15} for an alternative representation of the spreads in terms of Gaussian factors that guarantees the spreads to remain nonnegative and still allows for correlation between $r_t$ and $s_t$). Note, however, that given the current market situation where the observed interest rates are very close to zero and sometimes also negative, even models with negative mean-reversion level have been considered, as well as models allowing for regime-switching in the mean reversion parameter.
\end{rem}
\begin{rem}\label{R.2}  For the short rate itself one could also consider the
model $r_t=\phi_t+\Psi^1_t+(\Psi^2_t)^2$ where $\phi_t$ is a {\sl deterministic shift extension} (see \citet{BrigoMercurio06})
that allows for a good fit to the initial term structure in short
rate models even with constant model parameters.\end{rem}

In the model (\ref{8}) we have included a linear term $\Psi^1_t$
which may lead to negative values of rates and spreads, although
only with small probability in the case of models of the type
(\ref{7}) with a positive mean reversion level. The advantage of
including this linear term is more generality and flexibility in the
model. Moreover, it allows to express $\bar p(t,T)$ in terms of
$p(t,T)$ multiplied by a factor. This property will lead  to an {\sl
adjustment factor} by which one can express post-crisis quantities
in terms of corresponding pre-crisis quantities, see
\citet{MorinoRunggaldier14} in which this idea has been firstly
proposed in the context of exponentially affine short rate models
for multiple curves.

\subsection{Bond prices (OIS and Libor bonds)}\label{S.2.2}

In this subsection we derive explicit pricing formulas for the OIS
bonds $p(t,T)$ as defined in (\ref{3}) and  the fictitious Libor
bonds $\bar p(t,T)$ as defined in (\ref{71}). Thereby, $r_t$ and
$s_t$ are supposed to be given by (\ref{8}) with the factor
processes $\Psi^i_t$ evolving under the standard martingale measure
$Q$ according to (\ref{7}). Defining the matrices \beq\label{dynam}
\begin{array}{cc}
F=
\begin{bmatrix}
-b^1&0&0\\
0&-b^2&0\\
0&0&-b^3
\end{bmatrix}\>,\quad&
D=
\begin{bmatrix}
\sigma^1&0&0\\
0&\sigma^2&0\\
0&0&\sigma^3
\end{bmatrix}\end{array}\eeq
and considering the vector factor process
$\Psi_t:=[\Psi_t^1,\Psi_t^2,\Psi_t^3]'$ as well as the multivariate
Wiener process $W_t:=[w_t^1, w_t^2, w_t^3]'$, where $'$ denotes
transposition, the dynamics (\ref{7}) can be rewritten in synthetic
form as \beq\label{Psi} d\Psi_t=F\Psi_tdt+DdW_t.\eeq

Using results on exponential quadratic term structures (see
\citet{GombaniRunggaldier01}, \citet{Filipovic2002}), we have
\begin{align}
\label{p}
 \notag p(t,T)&= E^Q\Bigl\{e^{-\int_t^Tr_u du}\Big|\mathcal{F}_t\Bigr\}=E^Q\Bigl\{e^{-\int_t^T(\Psi^1_u+(\Psi^2_u)^2)du}\Big|\mathcal{F}_t\Bigr\}\\
&= \mathrm{exp}\Bigl[-A(t,T)-B'(t,T)\Psi_t-\Psi'_tC(t,T)\Psi_t\Bigr]
\end{align}
and, setting $R_t:=r_t + s_t$,
\begin{align}
\label{barp}
\notag  \bar p(t,T)&= E^Q\Bigl\{e^{-\int_t^T R_u du}\Big|\mathcal{F}_t\Bigr\}
=E^Q\Bigl\{e^{-\int_t^T((1+\kappa)\Psi^1_u+(\Psi^2_u)^2+(\Psi^3_u)^2)du}\Big|\mathcal{F}_t\Bigr\}\\
 &= \mathrm{exp}\Bigl[-\bar A(t,T)-\bar
B'(t,T)\Psi_t-\Psi'_t\bar C(t,T)\Psi_t\Bigr],
\end{align}
where $A(t,T)$, $\bar A(t,T)$, $B(t,T)$, $\bar B(t,T)$, $C(t,T)$ and
$\bar C(t,T)$ are scalar, vector and matrix-valued deterministic
functions to be determined. % and where by $B'$ we denote the transpose of a generic vector $B$.

For this purpose we recall the Heath-Jarrow-Morton (HJM) approach
for the case when $p(t,T)$ in (\ref{p}) represents the price of a
risk-free zero coupon bond. The HJM approach leads to the so-called
HJM drift conditions that impose conditions on the coefficients in
(\ref{p}) so that the resulting prices $p(t,T)$ do not imply
arbitrage possibilities. Since the risk-free bonds are traded, the
no-arbitrage condition is expressed by requiring $\frac{p(t,
T)}{B_t}$ to be a $Q-$martingale for $B_t$ defined as in (\ref{1}) and
it is exactly this martingality property to yield the drift
condition. In our case, $p(t,T)$ is the price of an OIS bond that is
not necessarily traded and in general does  not coincide with the price of 
a risk-free bond. However, whether the OIS bond is traded or not,
$\frac{p(t, T)}{B_t}$ is a $Q-$martingale by the very definition of
$p(t,T)$ in (\ref{p}) (see the first equality  in (\ref{p}))
and so we can follow the same HJM approach to obtain conditions on
the coefficients in (\ref{p}) also in our case.

For what concerns, on the other hand, the coefficients in
(\ref{barp}), recall that $\bar p(t, T)$ is a fictitious asset that
is not traded and thus is not subject to any no-arbitrage condition.
Notice, however, that by analogy to $p(t, T)$ in (\ref{p}), by its very definition given in the first equality in  
(\ref{barp}), $\frac{\bar p(t, T)}{\bar B_t}$  is a $Q-$martingale
for $\bar B_t$ given by $\bar B_t := \exp \int_0^t R_u du$. The two
cases $p(t, T)$ and $\bar p(t, T)$ can thus be treated in complete
analogy provided that we use for $\bar p(t, T)$ the numeraire $\bar
B_t$.

%As a consequence one can formally apply the HJM drift condition also
%for the coefficients of the fictitious instantaneous forward rate
%$\bar f(t,T):=-\frac{\partial}{\partial T}\log \bar p(t,T)$  that,
%again, translates into ODEs for the coefficients in the rightmost
%term in (\ref{barp}).

We shall next derive from the $Q-$martingality of $\frac{p(t,
T)}{B_t}$ and $\frac{\bar p(t, T)}{\bar B_t}$ conditions on the
coefficients in (\ref{p}) and (\ref{barp}) that correspond to the
classical HJM drift condition and lead thus to ODEs for these
coefficients. For this purpose we shall proceed by analogy to
section 2 in \citet{GombaniRunggaldier01}, in particular to the
proof of Proposition 2.1 therein, to which we also refer for more
detail.

%We have now to make sure that $Q$ can indeed be considered as a
%martingale measure, i.e. that the discounted values of $p(t,T)$ are
%local martingales and we shall require this property also for $\bar
%p(t,T)$.

% there are various possible approaches, here
%we follow [GR] thereby imposing a HJM drift condition after having
%we first derive from $p(t,T)$ and $\bar p(t,T)$ the corresponding
%instantaneous forward rates $f(t,T):=-\frac{\partial}{\partial
%T}\log p(t,T)$  and  $\bar f(t,T):=-\frac{\partial}{\partial T}\log \bar
%p(t,T).$

Introducing the ``instantaneous forward rates''
$f(t,T):=-\frac{\partial}{\partial T}\log p(t,T)$  and  $\bar
f(t,T):=-\frac{\partial}{\partial T}\log \bar p(t,T),$ and setting
\begin{equation}\label{coeff}
a(t,T):=\frac{\partial}{\partial T}A(t,T)\,,\quad
b(t,T):=\frac{\partial}{\partial T}B(t,T)\,,\quad
c(t,T):=\frac{\partial}{\partial T}C(t,T)
\end{equation}
and  analogously for $\bar a(t,T), \bar b(t,T), \bar c(t,T)$, from
(\ref{p}) and (\ref{barp}) we obtain
\begin{equation}\label{f}
f(t,T)=a(t,T)+b'(t,T)\Psi_t+\Psi'_tc(t,T)\Psi_t,
\end{equation}
\begin{equation}\label{barf}
\bar f(t,T)=\bar a(t,T)+\bar b'(t,T)\Psi_t+\Psi'_t\bar c(t,T)\Psi_t.
\end{equation}
Recalling that $r_t=f(t, t)$ and $R_t=\bar f(t, t)$, this implies,
with $a(t):=a(t,t), b(t):=b(t,t), c(t):=c(t,t)$ and analogously for
the corresponding quantities with a bar, that
\begin{equation}\label{r-quadratic}
r_t=a(t)+b'(t)\Psi_t+\Psi'_t c(t)\Psi_t
\end{equation}
and
\begin{equation}\label{s-quadratic}
R_t=r_t+s_t=\bar a(t)+\bar b'(t)\Psi_t+\Psi'_t\bar c(t)\Psi_t.
\end{equation}
Comparing (\ref{r-quadratic}) and (\ref{s-quadratic}) with
(\ref{8}), we obtain the following conditions where $i,j=1,2,3$,
namely
$$\begin{array}{cc}\begin{cases}
a(t)=0\\
b^i(t)=\Ind{i=1}\\
c^{ij}(t)=\Ind{i=j=2}
\end{cases}\quad& \begin{cases}
\bar a(t)=0\\
\bar b^i(t)=(1+\kappa)\Ind{i=1}\\
\bar c^{ij}(t)=\Ind{i=j=2\}\cup\{i=j=3}.
\end{cases}\end{array}$$

Using next the fact that
$$p(t,T)=\exp\left[-\displaystyle\int_t^Tf(t,s)ds\right]\,,\quad \bar p(t,T)=\exp\left[-\displaystyle\int_t^T\bar
f(t,s)ds\right],$$ and imposing $\frac{p(t, T)}{B_t}$ and
$\frac{\bar p(t, T)}{\bar B_t}$ to be $Q-$martingales, one obtains
ordinary differential equations to be satisfied by $c(t,T), b(t,T),
a(t,T)$  and analogously for the quantities with a bar. Integrating these ODEs with respect to the second variable
and recalling (\ref{coeff}) one obtains (for the details see the
proof of Proposition 2.1 in \citet{GombaniRunggaldier01})
\begin{equation}\label{system}
\begin{cases}
C_t(t,T)+2FC(t,T)-2C(t,T)DDC(t,T)+c(t)=0,\quad &C(T,T)=0\\
\bar C_t(t,T)+2F\bar C(t,T)-2\bar C(t,T)DD\bar C(t,T)+\bar c(t)=0,\quad &\bar C(T,T)=0\\
\end{cases}
\end{equation}
with
\begin{equation}
c(t)=
\begin{bmatrix}
0&0&0\\
0&1&0\\
0&0&0
\end{bmatrix}
\quad \bar c(t)=
\begin{bmatrix}
0&0&0\\
0&1&0\\
0&0&1
\end{bmatrix}.
\end{equation}
The special forms of $F$, $D$, $c(t)$ and $\bar c(t)$ together with
boundary conditions $ C(T,T)=0$ and $\bar C(T,T)=0 $ imply that only
$C^{22},\bar C^{22},\bar C^{33}$ are non-zero and satisfy
\begin{equation}\label{C^22}
\begin{cases}
C^{22}_t(t,T)-2b^2C^{22}(t,T)-2(\sigma^2)^2(C^{22}(t,T))^2+1=0,\quad& C^{22}(T,T)=0\\
\bar C^{22}_t(t,T)-2b^2\bar C^{22}(t,T)-2(\sigma^2)^2(\bar C^{22}(t,T))^2+1=0,\quad& \bar C^{22}(T,T)=0\\
\bar C^{33}_t(t,T)-2b^3\bar C^{33}(t,T)-2(\sigma^3)^2(\bar
C^{33}(t,T))^2+1=0,\quad& \bar C^{33}(T,T)=0
\end{cases}
\end{equation}
that can be shown to have as solution
\begin{equation}\label{eqdiff}
\begin{cases}
C^{22}(t,T)=\bar C^{22}(t,T)=\frac{2(e^{(T-t)h^2}-1)}{2h^2+(2b^2+h^2)(e^{(T-t)h^2}-1)}\\
\bar
C^{33}(t,T)=\frac{2(e^{(T-t)h^3}-1)}{2h^3+(2b^3+h^3)(e^{(T-t)h^3}-1)}
\end{cases}
\end{equation}
with $h^i=\sqrt{4(b^i)^2+8(\sigma^i)^2}>0,\,i=2,3.$

Next, always by analogy to the proof of Proposition 2.1 in
\citet{GombaniRunggaldier01}, the vectors of coefficients $B(t,T)$
and $\bar B(t,T)$ of the first order terms can be seen to satisfy
the following system
\begin{equation}\label{system2}
\begin{cases}
B_t(t,T)+B(t,T)F-2B(t,T)DDC(t,T)+b(t)=0,\quad&B(T,T)=0\\
\bar B_t(t,T)+\bar B(t,T)F-2\bar B(t,T)DD\bar C(t,T)+\bar
b(t)=0,\quad&\bar B(T,T)=0
\end{cases}
\end{equation}
with
$$b(t)=[1,0,0]\quad\bar b(t)=[(1+\kappa),0,0].$$
Noticing similarly as above that only $B^1(t,T),\bar B^1(t,T)$ are
non-zero, %and using \eqref{dynam},
system (\ref{system2}) becomes
\begin{equation}
\begin{cases}
B^1_t(t,T)-b^1B^1(t,T)+1=0\quad&B^1(T,T)=0\\
\bar B^1_t(t,T)-b^1\bar B^1(t,T)+(1+\kappa)=0\quad&\bar B^1(T,T)=0
\end{cases}
\end{equation}
leading to the explicit solution
\begin{equation}\label{B^1}
\begin{cases}
B^1(t,T)=\frac{1}{b^1}\Bigl(1 - e^{-b^1(T-t)}\Bigr)\\
\bar
B^1(t,T)=\frac{1+\kappa}{b^1}\Bigl(1 - e^{-b^1(T-t)}\Bigr)=(1+\kappa)B^1(t,T).
\end{cases}
\end{equation}
\\\\
Finally, $A(t,T)$ and $\bar A(t,T)$ have to satisfy
\begin{equation}\label{A}
\begin{cases}
A_t(t,T)+(\sigma^2)^2C^{22}(t,T)-\frac{1}{2}(\sigma^1)^2(B^1(t,T))^2=0,\\
\bar
A_t(t,T)+(\sigma^2)^2 \bar C^{22}(t,T)+(\sigma^3)^2 \bar  C^{33}(t,T)-\frac{1}{2}(\sigma^1)^2(\bar  B^1(t,T))^2=0
\end{cases}
\end{equation}
with boundary conditions $A(T,T)=0,\bar A(T,T)=0.$ The explicit expressions can be obtained simply by integrating the above equations.

Summarizing, we have proved the following
\begin{prop}
\label{p:bond-prices}
Assume that the OIS short rate $r$ and the spread $s$ are given by (\ref{8}) with the factor processes $\Psi^i_t$, $i=1, 2, 3$, evolving according to (\ref{7}) under the standard martingale measure $Q$. The time-$t$ price of the OIS bond $p(t,T)$, as defined in (\ref{3}), is given by
\begin{equation}\label{pt}
p(t,T)=\mathrm{exp}[-A(t,T)-B^1(t,T)\Psi^1_t-C^{22}(t,T)(\Psi^2_t)^2],
\end{equation}
and the time-$t$ price of the fictitious Libor bond  $\bar p(t,T)$, as defined in (\ref{71}), by
\begin{equation}\label{barp4}
\begin{split}
\bar p(t,T)&=\mathrm{exp}[-\bar A(t,T)-(\kappa+1)B^1(t,T)\Psi^1_t-C^{22}(t,T)(\Psi^2_t)^2-\bar C^{33}(t,T)(\Psi^3_t)^2]\\
&=p(t,T)\mathrm{exp}[-\tilde A(t,T)-\kappa B^1(t,T)\Psi^1_t-\bar
C^{33}(t,T)(\Psi^3_t)^2],
\end{split}
\end{equation}
where $\tilde A(t,T):=\bar A(t,T)-A(t,T)$ with  $A(t,T)$ and $\bar A(t,T)$ given by (\ref{A}), $B^1(t,T)$ given by (\ref{B^1}) and  $C^{22}(t,T)$ and $C^{33}(t,T)$ given by (\ref{eqdiff}).
\end{prop}

In particular, expression (\ref{barp4}) gives $\bar p(t,T)$ in terms of $p(t,T)$. Based on this  we shall derive   in the following section the announced {\sl adjustment
factor} allowing to pass from pre-crisis quantities  to the corresponding  post-crisis quantities.

\subsection{Forward measure}\label{S.3.1}

The underlying factor model was defined in (\ref{7}) under the
standard martingale measure $Q$. For derivative prices, which we
shall determine in the following two sections, it will be convenient
to work under forward measures, for which, given the single tenor
$\Delta$, we shall consider a generic $(T+\Delta)$-forward measure.
The density process to change the measure from $Q$ to $Q^{T+\Delta}$
is \beq\label{density}{\cal
L}_t:=\displaystyle\frac{d\,Q^{T+\Delta}}{d\,Q}
\Big|_{\F_t}=\frac{p(t,T+\Delta)}{p(0,T+\Delta)}\,\frac{1}{B_t}\eeq
from which it follows by (\ref{pt}) and the martingale property of
$\left(\frac{p(t, T+\Delta)}{B_t}\right)_{t\leq T+\Delta}$ that
$$d{\cal L}_t={\cal
L}_t\,\left(-B^1(t,T+\Delta)\sg^1dw_t^1-2C^{22}(t,T+\Delta)\Psi^2_t\sg^2dw_t^2\right).
$$
This implies by Girsanov's theorem that
\beq\label{9}\left\{\begin{array}{lcl}
dw_t^{1,T+\Delta}&=&dw_t^1+\sg^1B^1(t,T+\Delta)dt\\
dw_t^{2,T+\Delta}&=&dw_t^2+2C^{22}(t,T+\Delta)\Psi^2_t\sg^2dt\\
dw_t^{3,T+\Delta}&=&dw_t^3\end{array}\right.\eeq are
$Q^{T+\Delta}-$Wiener processes. From the $Q-$dynamics (\ref{7}) we
then obtain the following $Q^{T+\Delta}-$dynamics for the factors
\beq\label{10} \begin{array}{lcl} d\Psi_t^1&=&-\left[b^1
\Psi_t^1+(\sg^1)^2B^1(t,T+\Delta)\right]\,dt+\sg^1dw_t^{1,T+\Delta}\\
d\Psi_t^2&=&-\left[b^2
\Psi_t^2+2(\sg^2)^2C^{22}(t,T+\Delta)\Psi^2_t\right]\,dt+\sg^2dw_t^{2,T+\Delta}\\
d\Psi_t^3&=&-b^3 \Psi_t^3dt+\sg^3dw_t^{3,T+\Delta}.\end{array}\eeq
\begin{rem}\label{mean-rev}
While in the dynamics (\ref{7}) for $\Psi_t^i,\>(i=1,2,3)$ under $Q$
we had for simplicity assumed a zero mean-reversion level, under the
$(T+\Delta)$-forward measure the mean-reversion level is for
$\Psi_t^1$  now different from zero due to the measure
transformation.\end{rem}

\begin{lem}\label{mart} Analogously to the case when $p(t,T)$
represents the price of a risk-free zero coupon bond, also for
$p(t,T)$ viewed as OIS bond we have that
$\frac{p(t,T)}{p(t,T+\Delta)}$ is a
$Q^{T+\Delta}-$martingale.\end{lem}

\begin{proof} We have seen that also for OIS bonds as defined in
(\ref{3}) we have that, with $B_t$ as in (\ref{1}), the ratio
$\frac{p(t,T)}{B_t}$ is a $Q-$martingale. From Bayes' formula we
then have
$$\begin{array}{l}
E^{T+\Delta}\left\{\frac{p(T,T)}{p(T,T+\Delta)}\mid\F_t\right\}
=\frac{E^{Q}\left\{\frac{1}{p(0,T+\Delta)}\frac{1}{B_{T+\Delta}}\frac{p(T,T)}{p(T,T+\Delta)}\mid\F_t\right\}}
{E^{Q}\left\{\frac{1}{p(0,T+\Delta)}\frac{1}{B_{T+\Delta}}\mid\F_t\right\}}\\
\\
\quad=\frac{E^{Q}\left\{\frac{p(T,T)}{p(T,T+\Delta)}E^Q\left\{\frac{1}{B_{T+\Delta}}\mid\F_T\right\}\mid\F_t\right\}}
{\frac{p(t,T+\Delta)}{B_t}}=\frac{B_t
E^{Q}\left\{\frac{p(T,T)}{p(T,T+\Delta)}\frac{p(T,T+\Delta)}{B_{T}}\mid\F_t\right\}}
{p(t,T+\Delta)}\\ \\\quad=\frac{B_t
E^{Q}\left\{\frac{p(T,T)}{B_T}\mid\F_t\right\}}
{p(t,T+\Delta)}=\frac{p(t,T)}{p(t,T+\Delta}),\end{array}$$ thus
proving the statement of the lemma. \qed  \end{proof}

We recall that we denote the expectation with respect to the measure
$Q^{T+\Delta}$ by $E^{T+\Delta}\{\cdot\}$.  The dynamics in
(\ref{10}) lead to Gaussian distributions for $\Psi_t^i,\>i=1,2,3$
that, given $B^1(\cdot)$ and $C^{22}(\cdot)$, have mean and variance
$$E^{T+\Delta}\{\Psi_t^i\}=\bar\al_t^i=\bar\al_t^i(b^i,\sg^i)\quad,\quad
Var^{T+\Delta}\{\Psi_t^i\}=\bar\beta_t^i=\bar\beta_t^i(b^i,\sg^i),$$
which can be explicitly computed. More precisely, we have
\beq\label{11}
\begin{cases}
\bar\alpha_t^1&=e^{-b^1t}\Bigl[\Psi^1_0-\frac{(\sigma^1)^2}{2(b^1)^2}e^{-b^1(T+\Delta)}(1-e^{2b^1t})-\frac{(\sigma^1)^2}
{(b^1)^2}(1-e^{b^1t})\Bigr]\\
\bar\beta_t^1&=e^{-2b^1t}(e^{2b^1t}-1)\frac{(\sigma^1)^2}{2(b^1)}\\
\bar\alpha_t^2&=e^{-(b^2t+2(\sigma^2)^2\tilde C^{22}(t,T+\Delta))}\Psi^2_0\\
\bar\beta^2_t&=e^{-(2b^2t+4(\sigma^2)^2\tilde C^{22}(t,T+\Delta))}\int_0^te^{2b^2s+4(\sigma^2)^2\tilde C^{22}(s,T+\Delta)}(\sigma^2)^2ds\\
\bar\alpha_t^3&=e^{-b^3t}\Psi^3_0\\
\bar\beta^3_t&=e^{-2b^3t}\frac{(\sigma^3)^2}{2b^3}(e^{2b^3t}-1),
\end{cases}
\eeq with \beq\label{12}
\begin{split}
\tilde C^{22}(t,T+\Delta)&=\frac{2(2\log(2b^2(e^{(T+\Delta-t)h^2}-1)+h^2(e^{(T+\Delta-t)h^2}+1))+t(2b^2+h^2))}{(2b^2+h^2)(2b^2-h^2)}\\
&-\frac{2(2\log(2b^2(e^{(T+\Delta)h^2}-1)+h^2(e^{(T+\Delta)h^2}+1))}{(2b^2+h^2)(2b^2-h^2)}
\end{split}
\eeq and $h^2=\sqrt{(2b^2)^2+8(\sigma^2)^2}$, and where we have
assumed deterministic initial values $\Psi^1_0,\Psi^2_0$ and $
\Psi^3_0.$  For details of the above computation see the proof of
Corollary 4.1.3. in \citet{Meneghello14}.

\section{Pricing of linear interest rate derivatives}\label{S.3}

We have discussed in subsection \ref{S.2.2} the pricing of OIS and
Libor bonds in the Gaussian, exponentially quadratic short rate
model introduced in subsection \ref{S.2.1}. In the remaining part of
the paper we shall be concerned with the pricing of interest rate
derivatives, namely with derivatives having the Libor rate as
underlying rate. In the present section we shall deal with the basic
linear derivatives, namely FRAs and interest rate swaps, while
nonlinear derivatives will then be dealt with in the following
section \ref{S.4}. For the FRA rates discussed in the next
subsection \ref{S.3.2} we shall in sub-subsection \ref{S.3.2.a}
exhibit an {\sl adjustment factor} allowing to pass from the
single-curve FRA rate to the multi-curve FRA rate.

\subsection{FRAs}\label{S.3.2}

We start by recalling the definition of a standard forward rate
agreement. We emphasize that we use a text-book definition which
differs slightly from a market definition, see \citet{Mercurio10a}.
\begin{df}\label{FRA}
Given the time points $0\leq t\le T<T+\Delta$, a {\sl forward rate
agreement} (FRA) is an OTC derivative that allows the holder to lock
in at the generic date $t\leq T$ the interest rate between the inception
date $T$ and the maturity $T+\Delta$ at a fixed value $R$. At
maturity $T+\Delta$ a payment based on the interest rate $R$,
applied to a notional amount of $N$, is made and the one based on the
relevant floating rate {\sl (generally the spot Libor rate
$L(T;T,T+\Delta)$)} is received.
\end{df}
Recalling that for the Libor rate we had postulated the relation
(\ref{6}) to hold at the inception time $T$, namely
$$L(T;T,T+\Delta)=\displaystyle\frac{1}{\Delta}\,\left(\frac{1}{\bar
p(T,T+\Delta)}-1\right),$$ the price, at $t\leq T,$ of the FRA with
fixed rate $R$ and notional $N$ can be computed under the
$(T+\Delta)-$ forward measure as \beq\label{PFRA}\begin{array}{l}
P^{FRA}(t;T,T+\Delta,R,N)\\ \\ \quad=N\Delta
p(t,T+\Delta)E^{T+\Delta}\left\{L(T;T,T+\Delta)-R\mid\F_t\right\}\\
\\ \quad =N p(t,T+\Delta)E^{T+\Delta}\left\{\frac{1}{\bar
p(T,T+\Delta)}-(1+\Delta R)\mid\F_t\right\},\end{array}\eeq
%Since we consider here a single tenor $\Delta$,
Defining \beq\label{nubar}
\bar\nu_{t,T}:=E^{T+\Delta}\left\{\frac{1}{\bar
p(T,T+\Delta)}\mid\F_t\right\},\eeq it is easily seen from
(\ref{PFRA}) that the {\sl fair rate of the FRA, namely the FRA
rate}, is given by \beq\label{FRARbar} \bar
R_t=\frac{1}{\Delta}\left(\bar\nu_{t,T}-1\right).\eeq
In the {\sl single-curve case} we have instead \beq\label{FRAR}
R_t=\frac{1}{\Delta}\left(\nu_{t,T}-1\right),\eeq where, given that
$\frac{p(\cdot,T)}{p(\cdot,T+\Delta)}$ is a
$Q^{T+\Delta}-$martingale (see Lemma \ref{mart}),
\beq\label{nu}\nu_{t,T}:=E^{T+\Delta}\left\{\frac{1}{p(T,T+\Delta)}\mid\F_t\right\}=\frac{p(t,T)}{p(t,T+\Delta)},\eeq
which is the classical expression for the FRA rate in the single-curve case. Notice that, contrary to (\ref{nubar}), the expression
in (\ref{nu}) can be explicitly computed on the basis of bond price
data without requiring an interest rate model.

\subsubsection{Adjustment factor}\label{S.3.2.a}

We shall show here the following
\begin{prop}\label{P.3.1} We have the relationship
\beq\label{adj}\bar\nu_{t,T}=\nu_{t,T}\cdot Ad_t^{T,\Delta}\cdot
Res_t^{T,\Delta}\eeq with
\beq\label{Adj1}\begin{array}{l}Ad_t^{T,\Delta}:=E^Q\left\{\frac{p(T,T+\Delta)}{\bar
p(T,T+\Delta)}\mid\F_t\right\}=E^Q\Bigl\{\exp\Bigl[\tilde
A(T,T+\Delta)\\ \hspace{1.5cm}+\kappa B^1(T,T+\Delta)\Psi_T^1+\bar
C^{33}(T,T+\Delta)(\Psi_T^3)^2\Bigr]\mid\F_t\Bigr\}\end{array}\eeq
and
\beq\label{Adj2}Res_t^{T,\Delta}=\exp\Bigl[-\kappa\frac{(\sg^1)^2}{2(b^1)^3}\left(1-e^{-b^1\Delta}\right)
\left(1-e^{-b^1(T-t)}\right)^2\Bigr],\eeq where $\tilde A(t,T)$ is
defined after (\ref{barp4}), $B^1(t,T)$ in (\ref{B^1}) and $\bar
C^{33}(t,T)$ in (\ref{eqdiff}).
\end{prop}

\begin{proof}
Firstly, from (\ref{barp4}) we obtain
\begin{equation}\label{p/barp2}
\frac{p(T,T+\Delta)}{\bar p(T,T+\Delta)}=e^{\tilde
A(T,T+\Delta)+\kappa B^1(T,T+\Delta)\Psi^1_{T}+\bar
C^{33}(T,T+\Delta)(\Psi^3_{T})^2}.
\end{equation}
In (\ref{nubar}) we now change back from the $(T+\Delta)-$ forward measure to the standard
martingale measure using the density process ${\cal L}_t$
given in (\ref{density}).  Using furthermore the above expression for the ratio
of the OIS and the Libor bond prices and 
taking into account the definition of the short rate $r_t$ in
terms of the factor processes, we obtain \beq\label{P1}
\begin{split}
\bar{\nu}_{t,T}&=E^{T+\Delta}\biggl\{\frac{1}{\bar{p}(T,T+\Delta)}\big|\mathcal{F}_t\biggr\}={\cal
L}_t^{-1}E^Q\biggl\{\frac{{\cal L}_{T}}
{\bar{p}(T,T+\Delta)}\big|\mathcal{F}_t\biggr\}\\
&=\frac{1}{p(t,T+\Delta)}E^Q\biggl\{\mathrm{exp}\Bigl(-\int_t^{T}r_u du\Bigr)\frac{p(T,T+\Delta)}{\bar{p}(T,T+\Delta)}\big|\mathcal{F}_t\biggr\}\\
&=\frac{1}{p(t,T+\Delta)}\mathrm{exp}[\tilde A(T,T+\Delta)]E^Q\Bigl\{e^{\bar C^{33}(T,T+\Delta)(\Psi^3_{T})^2}\big|\mathcal{F}_t\Bigr\}\\
&\hspace{2cm}\cdot
E^Q\Bigl\{e^{-\int_t^{T}(\Psi^1_u+(\Psi^2_u)^2)du}e^{\kappa
B^1(T,T+\Delta)\Psi^1_T}\big|\mathcal{F}_t\Bigr\}\\
&=\frac{1}{p(t,T+\Delta)}\mathrm{exp}[\tilde
A(T,T+\Delta)]E^Q\Bigl\{e^{\bar
C^{33}(T,T+\Delta)(\Psi^3_{T})^2}\big|\mathcal{F}_t\Bigr\}\\
&\hspace{.8cm}\cdot E^Q\Bigl\{e^{-\int_t^{T}\Psi^1_udu}e^{\kappa
B^1(T,T+\Delta)\Psi^1_{T}}\big|\mathcal{F}_t\Bigr\}E^Q\Bigl\{e^{-\int_t^{T}(\Psi^2_u)^2du}
\big|\mathcal{F}_t\Bigr\},
\end{split}
\eeq where we have used the independence  of the factors
$\Psi^i,\>i=1,2,3$ under $Q$.

Recall now from the theory of affine processes (see e.g. Lemma 2.1
in \citet{GrbacRunggaldier15}) that, for a process $\Psi_t^1$ satisfying (\ref{7}), we
have for all $\delta, K \in \R$
 \beq\label{affine}E^Q\left\{\exp\left[-\int_t^T\delta
\Psi_u^1du-K\Psi_T^1\right]\mid\F_t\right\}=\exp[\alpha^1(t,T)-\beta^1(t,T)\Psi_t^1],\eeq
where
$$\left\{\begin{array}{lcl}
\beta^1(t,T)&=&Ke^{-b^1(T-t)}-\frac{\delta}{b^1}\left(e^{-b^1(T-t)}-1\right)\\
\alpha^1(t,T)&=&\frac{(\sg^1)^2}{2}\,\int_t^T(\beta^1(u,T))^2du.\end{array}\right.$$
Setting $K=-\kappa \,B^1(T,T+\Delta)$ and $\delta=1$, and recalling
from (\ref{B^1}) that
$B^1(t,T)=\frac{1}{b^1}\Bigl( 1 - e^{-b^1(T-t)}\Bigr)$, this leads to
\beq\label{T1}\begin{array}{l}
E^Q\Bigl\{e^{-\int_t^{T}\Psi^1_udu}e^{\kappa
B^1(T,T+\Delta)\Psi^1_{T}}\big|\mathcal{F}_t\Bigr\}\\
=\mathrm{exp}\Biggl[\frac{(\sigma^1)^2}{2}(\kappa
B^1(T,T+\Delta))^2\displaystyle\int_t^{T}e^{-2b^1(T-u)}du \\
\qquad -\kappa
B^1(T,T+\Delta)(\sigma^1)^2\displaystyle\int_t^{T} B^1(u,T)e^{-b^1(T-u)}du
+ \frac{(\sigma^1)^2}{2}\int_t^{T}(B^1(u,T))^2du
\\
\qquad+\left(\kappa
B^1(T,T+\Delta)e^{-b^1(T-t)}-B^1(t,T)\right)\Psi^1_t\Biggr].\end{array}\eeq
On the other hand, from the results of section \ref{S.2.2} we also
have that, for a process $\Psi_t^2$ satisfying (\ref{7}),
$$E^Q\left\{\exp\left[-\displaystyle\int_t^T(\Psi_u^2)^2du\right]\mid\F_t\right\}=\exp\left[-\alpha^2(t,T)-C^{22}(t,T)(\Psi_t^2)^2\right],$$
where $C^{22}(t,T)$ corresponds to (\ref{eqdiff}) and (see
(\ref{A}))
$$\alpha^2(t,T)=(\sg^2)^2\int_t^TC^{22}(u,T)du.%-\frac{(\sigma^1)^2}{2}\int_t^{T}(B^1(u,T))^2du.
$$
This implies that \beq\label{T2}\begin{array}{l}
E^Q\left\{\exp\left[-\displaystyle\int_t^T(\Psi_u^2)^2du\right]\mid\F_t\right\}\\
\>\>=\exp\left[-(\sg^2)^2\displaystyle\int_t^TC^{22}(u,T)du %+\frac{(\sigma^1)^2}{2}\int_t^{T}(B^1(u,T))^2du
-C^{22}(t,T)\left(\Psi_t^2\right)^2\right].\end{array}\eeq Replacing
(\ref{T1}) and (\ref{T2}) into (\ref{P1}), and recalling the
expression for $p(t,T)$ in (\ref{pt}) with $A(\cdot), B^1(\cdot),
C^{22}(\cdot)$ according to (\ref{A}), (\ref{B^1}) and
(\ref{eqdiff}) respectively, we obtain \beq\label{P9}
\begin{array}{l}
\bar{\nu}_{t,T}=\frac{p(t,T)}{p(t,T+\Delta)}e^{\tilde
A(T,T+\Delta)}E^Q\Bigl[e^{\bar C^{33}(T,T+\Delta)
(\Psi^3_{T})^2}\big|\mathcal{F}_t\Bigr]\\
\quad\cdot \mathrm{exp}\Bigl[\frac{(\sigma^1)^2}{2}(\kappa
B^1(T,T+\Delta))^2\displaystyle\int_t^{T}e^{-2b^1(T-u)}du+
\kappa B^1(T,T+\Delta)e^{-b^1(T-t)}\Psi^1_t\Bigr]\\
\quad\cdot\mathrm{exp}\Bigl[-\kappa
B^1(T,T+\Delta)(\sigma^1)^2\displaystyle\int_t^{T}B^1(u,T)e^{-b^1(T-u)}du\Bigr].
\end{array}
\eeq We recall  the expression (\ref{p/barp2}) for
$\frac{p(T,T+\Delta)}{\bar p(T,T+\Delta)}$ and the fact that,
according to (\ref{affine}), we have
$$\begin{array}{l}
E^Q\Bigl\{e^{\kappa
B^1(T,T+\Delta)\Psi^1_{T}}\big|\mathcal{F}_t\Bigr\}\\
\quad=\exp\left[\frac{(\sigma^1)^2}{2}(\kappa
B^1(T,T+\Delta))^2\displaystyle\int_t^{T}e^{-2b^1(T-u)}du+\kappa
B^1(T,T+\Delta)e^{-b^1(T-t)}\Psi_t^1\right].\end{array}$$
Inserting these expressions  into (\ref{P9}) we obtain  the result, namely \beq\label{P10}
\begin{array}{lcl}
\bar{\nu}_{t,T}&=&
\frac{p(t,T)}{p(t,T+\Delta)}E^Q\Bigl\{\frac{p(T,T+\Delta)}{\bar
p(T,T+\Delta)}\big|\mathcal{F}_t\Bigr\} \\
&{}&  \cdot \mathrm{exp}\Bigl[-\kappa B^1(T,T+\Delta)(\sigma^1)^2\displaystyle\int_t^{T}B^1(u,T)e^{-b^1(T-u)}du\Bigr]\\
&=&\frac{p(t,T)}{p(t,T+\Delta)}E^Q\Bigl\{\frac{p(T,T+\Delta)}{\bar p(T,T+\Delta)}\big|\mathcal{F}_t\Bigr\}\\
&{}&  \cdot \mathrm{exp}\Bigl[-\frac{\kappa}{b^1}(e^{-b^1\Delta}-1)(\sigma^1)^2\Bigl(\frac{1}{2(b^1)^2}
(1-e^{-2b^1(T-t)})-\frac{1}{(b^1)^2}(1-e^{-b^1(T-t)})\Bigr)\Bigr],
\end{array}
\eeq where we have also used the fact that
$$\begin{array}{l}
\int_t^T B^1(u,T)e^{-b^1(T-u)}du=\displaystyle\int_t^T
\frac{1}{b^1}\left(1 - e^{-b^1(T-u)}\right)e^{-b^1(T-u)}du\\
\hspace{3cm}=- \frac{1}{2(b^1)^2}\left(1-e^{-2b^1(T-t)}\right)+\frac{1}{(b^1)^2}\left(1-e^{-b^1(T-t)}\right). \\   \hspace{11cm}   \qed \end{array} $$
\end{proof}

\begin{rem}\label{R.3.1} The adjustment factor $Ad_t^{T,\Delta}$
allows for some intuitive interpretations. Here we mention only the
easiest one for the case when $\kappa=0$ (independence of $r_t$ and
$s_t$). In this case we have $r_t+s_t>r_t$ implying that $\bar
p(T,T+\Delta)<p(T,T+\Delta)$ so that $Ad_t^{T,\Delta}\ge 1$.
Furthermore, always for $\kappa=0$, the residual factor has value
$Res_t^{T,\Delta}=1$. All this in turn implies $\bar\nu_{t,T}\ge
\nu_{t,T}$ and with it $\bar R_t\ge R_t$, which is what one would
expect to be the case.\end{rem}

\begin{rem}\label{R.3.2} ({\sl Calibration to the initial term structure}). The parameters in the
model (\ref{7}) for the factors $\Psi_t^i$ and thus also in the
model (\ref{8}) for the short rate $r_t$ and the spread $s_t$ are
the coefficients $b^i$ and $\sg^i$ for $i=1,2,3.$ From (\ref{p})
notice that, for $i=1,2$, these coefficients enter the expressions
for the OIS bond prices $p(t,T)$ that can be assumed to be
observable since they can be bootstrapped from the market quotes for the OIS swap rates. % as in the classical short rate models.
We may thus assume
that these coefficients, i.e. $b^i$ and $\sg^i$ for $i=1,2$, can be
calibrated as in the pre-crisis single-curve short rate models.
It remains to
calibrate $b^3$, $\sg^3$ and, possibly the correlation coefficient $\kappa$. Via (\ref{barp})
they affect the prices of the fictitious Libor bonds $\bar p(t,T)$
that are, however, not observable. One may observe though the FRA
rates $R_t$ and $\bar R_t$ and thus also $\nu_{t,T}$, as well as
$\bar\nu_{t,T}$. Via (\ref{adj}) this would then allow one to
calibrate also the remaining parameters. This task would turn out to
be even simpler if one would have access to the value of $\kappa$ by
other means.

We emphasize that in order to ensure a good fit to the initial bond
term structure, a deterministic shift extension of the model or
time-dependent coefficients  $b^i$ could be considered. We recall
also that we have assumed the mean-reversion level equal to zero for
simplicity;  in practice it would be one more coefficient to be
calibrated for each factor $\Psi_t^i$.
\end{rem}

\subsection{Interest rate swaps}\label{S.3.5}

We first recall the notion of a {\sl (payer) interest rate swap}.
Given a collection of dates $0 \leq T_0 < T_1 < \cdots < T_n$ with
$\gm\equiv\gm_{k} := T_{k} - T_{k-1}\>(k=1,\cdots,n)$, as well as a
notional amount $N$, a payer swap is a financial contract, where a
stream of interest payments on the notional $N$ is made at a fixed
rate $R$ in exchange for receiving an analogous stream corresponding
to the Libor rate. Among the various possible conventions concerning
the fixing for the Libor and the payment dates, we choose here the
one where, for each interval $[T_{k-1},T_k]$, the Libor rates are
fixed in advance and the payments are made in arrears. The swap is
thus initiated at $T_0$ and the first payment is made at $T_1$. A
{\sl receiver swap} is completely symmetric with the interest at the
fixed rate being received; here we concentrate on payer swaps.

The arbitrage-free price of the swap, evaluated at $t\leq T_0$, is
given by the following expression where, analogously to 
$E^{T+\Delta}\{\cdot\},$ we denote by $E^{T_k}\{\cdot\}$ the
expectation with respect to the forward measure
$Q^{T_k}\>(k=1,\cdots,n)$
\begin{eqnarray}\label{swa}
\nonumber P^{Sw}(t; T_0, T_n, R) & = & \gm\sum_{k=1}^{n} p(t, T_k)
E^{T_k}
\left\{L(T_{k-1}; T_{k-1}, T_k) - R | \mathcal{F}_t \right\} \\
& = &  \gm\sum_{k=1}^{n} p(t, T_k) \left( L(t; T_{k-1}, T_k) - R
\right).
\end{eqnarray}
For easier notation we have assumed the notional to be $1$, i.e.
$N=1$.

We shall next obtain an explicit expression for $P^{Sw}(t;T_0, T_n,
R)$ starting from the first equality in (\ref{swa}). To
this effect, recalling from (\ref{eqdiff}) that $C^{22}(t,T)=\bar
C^{22}(t,T)$, introduce again some shorthand notation, namely
\beq\label{shorth}
\begin{array}{l}
A_k:=\bar A(T_{k-1},T_{k}), B^1_k:=B^1(T_{k-1},T_{k}),\\
C^{22}_k:=C^{22}(T_{k-1},T_{k})=\bar C^{22}(T_{k-1},T_{k}), \>\bar
C^{33}_k:=\bar C^{33}(T_{k-1},T_{k}).\end{array}\eeq

The crucial quantity to be computed in (\ref{swa}) is the following
one
\begin{equation}\label{manip5}
\begin{split}
&E^{T_k}\{\gamma L(T_{k-1};T_{k-1},T_{k})|\mathcal{F}_{t}\}=E^{T_k}\Bigl\{\frac{1}{\bar p(T_{k-1},T_{k})}|\mathcal{F}_{t}\Bigr\}-1\\
&=e^{A_k}E^{T_k}\{\text{exp}((\kappa+1) B^1_k\Psi^1_{T_{k-1}}+C^{22}_k(\Psi^2_{T_{k-1}})^2+\bar C^{33}_k(\Psi^3_{T_{k-1}})^2)|\mathcal{F}_{t}\}-1,\\
\end{split}
\end{equation}
where we have used the first relation on the right in (\ref{barp4}).
The expectations in (\ref{manip5}) have to be computed under the
measures $Q^{T_k}$, under which, by analogy to (\ref{10}), the
factors have  the dynamics \beq\label{101}
\begin{array}{lcl} d\Psi_t^1&=&-\left[b^1
\Psi_t^1+(\sg^1)^2B^1(t,T_k)\right]\,dt+\sg^1dw_t^{1,k}\\
d\Psi_t^2&=&-\left[b^2
\Psi_t^2+2(\sg^2)^2C^{22}(t,T_k)\Psi^2_t\right]\,dt+\sg^2dw_t^{2,k}\\
d\Psi_t^3&=&-b^3 \Psi_t^3dt+\sg^3dw_t^{3,k}.\end{array}\eeq where
$w^{i,k}$, $i=1, 2, 3$, are independent Wiener processes with
respect to $Q^{T_k}$. A straightforward generalization of
(\ref{affine}) to the case where the factor process $\Psi_t^1$
satisfies the foll\-ow\-ing affine Hull-White model
$$d\Psi_t^1=(a^1(t)-b^1\Psi_t^1)dt+\sg^1dw_t$$
can be obtained as follows
\beq\label{affine1}E^Q\left\{\exp\left[-\int_t^T\delta
\Psi_u^1du-K\Psi_T^1\right]\mid\F_t\right\}=\exp[\al^1(t,T)-\beta^1(t,T)\Psi_t^1],\eeq
with \beq\label{HW}\left\{\begin{array}{lcl}
\beta^1(t,T)&=&Ke^{-b^1(T-t)}-\frac{\delta}{b^1}\left(e^{-b^1(T-t)}-1\right)\\
\al^1(t,T)&=&\frac{(\sg^1)^2}{2}\,\displaystyle\int_t^T(\beta^1(u,T))^2du-
\displaystyle\int_t^Ta^1(u)\beta^1(u,T)du.\end{array}\right.\eeq We
apply this result to our situation where under $Q^{T_k}$  the
process $\Psi_t^1$ satisfies the first SDE in (\ref{101}) and thus
corresponds to the above dynamics with
$a^1(t)=-(\sg^1)^2B^1(t,T_k).$ Furthermore, setting %as for (\ref{T1}),
$K=-(\kappa+1) \,B^1_k$ and $\delta=0$, we obtain for the first
expectation in the second line of (\ref{manip5}) \beq\label{T11}
E^{T_k}\{\text{exp}((\kappa+1)
B^1_k\Psi^1_{T_{k-1}}|\mathcal{F}_{t}\}=\exp[\Gamma^1(t,T_k)-\rho^1(t,T_k)\,\Psi^1_{t}],\eeq
with \beq\label{T12} \left\{\begin{array}{lcl}
\rho^1(t,T_k)&=&-(\kappa+1)B_k^1e^{-b^1(T_k-t)} \\ %-\frac{1}{b^1}\left(e^{-b^1(T_k-t)}-1\right)\\
\Gamma^1(t,T_k)&=&\frac{(\sg^1)^2}{2}\displaystyle\int_{t}^{T_k}\left(\rho^1(u,T_k)\right)^2du
+(\sg^1)^2\displaystyle\int_{t}^{T_k}B^1(u,T_k)\rho^1(u,T_k)du.
\end{array}\right.\eeq
For the remaining two expectations  in the second line of (\ref{manip5}) we
shall use the following

\begin{lem}\label{Laffine}
Let a generic process $\Psi_t$ satisfy the dynamics \beq\label{L1}
d\Psi_t=b(t)\Psi_tdt+\sg\,dw_t\eeq with $w_t$ a Wiener process.
Then, for all $C \in \R$ such that
$E^Q\left\{\exp\left[C\,(\Psi_T)^2\right]\right\} < \infty$, we have
\beq\label{L2}
E^Q\left\{\exp\left[C\,(\Psi_T)^2\right]\mid\F_t\right\}=\exp\left[\Gamma(t,T)-\rho(t,T)\,(\Psi_t)^2\right]\eeq
with $\rho(t,T)$ and $\Gamma(t,T)$ satisfying \beq\label{L3}
\left\{\begin{array}{l} \rho_t(t,T)+ 2b(t) \rho(t,T)-
2(\sg)^2\,(\rho(t,T))^2=0\>;\quad
\rho(T,T)=-C\\
\Gamma_t(t,T)=(\sg)^2 \rho(t,T).\end{array}\right.\eeq
%where \beq\label{L4}
%\ld(t)=2b(t)\quad,\quad \eta=2\sg\quad,\quad a=(\sg)^2\eeq
\end{lem}
\begin{proof} An application of It\^o's formula yields that the nonnegative process $\Phi_t:=(\Psi_t)^2$ satisfies the following SDE
$$ d\Phi_t=\left((\sg)^2 + 2b(t)\,\Phi_t\right)dt+ 2\sg
\,\sqrt{\Phi_t}\,dw_t. $$ We recall that a process $\Phi_t$ given in
general form by $$
d\Phi_t=(a+\ld(t)\Phi_t)dt+\eta\,\sqrt{\Phi_t}\,dw_t,$$ with $a, \eta
>0$ and $\ld(t)$ a deterministic function, is a CIR process. Thus,
$(\Psi_t)^2$ is equivalent in distribution to a CIR process with
coefficients given by $$ \ld(t)=2b(t)\quad,\quad
\eta=2\sg\quad,\quad a=(\sg)^2. $$ From the theory of affine term
structure models (see e.g. \citet{LambertonLapeyre07}, or Lemma 2.2
in \citet{GrbacRunggaldier15}) it now follows that
 %and we have
%that for a process $\Phi_t$ satisfying the CIR model \beq\label{L5}
%d\Phi_t=(a+\ld(t)\Phi_t)dt+\eta\,\sqrt{\Phi_t}\,dw_t\eeq one has
%$$E^Q\left\{\exp\left[C\,\Phi_T\right]\mid\F_t\right\}=\exp\left[\Gamma(t,T)-\rho(t,T)\, \Phi_t\right]$$
%with $\rho(t,T)$ and $\Gamma(t,T)$ satisfying (\ref{L3}).

%Hence, $(\Psi_t)^2$ is equivalent in distribution to a CIR process with coefficients given by (\ref{L4}).
%Next consider the process $\Psi_t$ satisfying (\ref{L1}). By It\^o's
%formula
%$$d(\Psi_t)^2=\left(2b(t)\,(\Psi_t)^2+\sg^2\right)dt+2\Psi_t\sg\,dw_t$$
%so that $\Phi_t$ in (\ref{L5}) can be seen as corresponding to
%$(\Psi_t)^2$ via the correspondence of the coefficients in
%(\ref{L4}).
%It then follows that
\begin{align*}
E^Q\left\{\exp\left[C\,(\Psi_T)^2\right]\mid\F_t\right\} & =E^Q\left\{\exp\left[C\,\Phi_T\right]\mid\F_t\right\} = \exp\left[\Gamma(t,T)-\rho(t,T)\, \Phi_t\right] \\
& =\exp\left[\Gamma(t,T)-\rho(t,T)\,(\Psi_t)^2\right]
\end{align*}
 with
$\rho(t,T)$ and $\Gamma(t,T)$ satisfying (\ref{L3}). \qed
\end{proof}
\begin{cor}\label{Caffine}
When $b(t)$ is constant with respect to time, i.e. $b(t)\equiv b$,
so that also $\ld(t)\equiv \ld$, then the equations for $\rho(t,T)$
and $\Gamma(t,T)$ in (\ref{L3}) admit an explicit solution given by
\beq\label{L6} \left\{\begin{array}{l} \rho(t,T)=
\frac{4bhe^{2b(T-t)}}{4(\sigma)^2he^{2b(T-t)}-1}\quad\mbox{with}\quad
h:=\frac{C}{4( \sg)^2C + 4b}\\
\Gamma(t,T)=-(\sg)^2\displaystyle\int_t^T\rho(u,T)du.\end{array}\right.\eeq
\end{cor}
Coming now to the second expectation in the second line of 
(\ref{manip5}) and using the second equation in (\ref{101}), we set
$$b(t):=-\left[b^2+2(\sg^2)^2C^{22}(t,T_k)\right],
\>\sg:=\sg^2,\>C=C_k^{22}$$ and apply Lemma \ref{Laffine}, provided
that the parameters  $b^2$ and $\sigma^2$ of the process $\Psi^2$
are such that $C=C_k^{22}$ satisfies the assumption from the lemma.
 We thus obtain
    \beq\label{L7}E^{T_k}\{\text{exp}(C^{22}_k(\Psi^2_{T_{k-1}})^2)|\mathcal{F}_{t}\}=\mathrm{exp}[\Gamma^2(t,T_k)-\rho^2(t,T_k)(\Psi^2_{t})^2],\eeq
with $\rho^2(t,T), \Gamma^2(t,T)$ satisfying
\beq\label{L8}\left\{\begin{array}{l}\rho^2_t(t,T)-2\left[b^2+2(\sg^2)^2C^{22}(t,T_k)\right]\rho^2(t,T)
-2(\sg^2)^2(\rho^2(t,T))^2=0\\ \rho^2(T_k,T_k)=-C_k^{22} \\
\Gamma^2(t,T)=-(\sg^2)^2\displaystyle\int_t^T \rho^2(u,T)du.
\end{array}\right.\eeq
Finally, for the third expectation in the second line of  (\ref{manip5}),
we may take advantage of the fact that the dynamics of $\Psi_t^3$ do
not change when passing from the measure $Q$ to the forward measure
$Q^{T_k}$. We can then apply Lemma \ref{Laffine}, this time with
(see the third equation in (\ref{101}))
$$b(t):=-b^3,
\>\sg:=\sg^3,\>C=\bar C_k^{33}$$ and  ensuring that the parameters
$b^3$ and $\sigma^3$ of the process $\Psi^3$ are such that $C=\bar
C_k^{33}$ satisfies the assumption from the lemma.
 Since $b(t)$ is constant with
respect to time, also Corollary \ref{Caffine}  applies and we obtain
$$E^{T_k}\{\text{exp}(\bar C^{33}_k(\Psi^3_{T_{k-1}})^2)|\mathcal{F}_{t}\}=\mathrm{exp}[\Gamma^3(t,T_k)-\rho^3(t,T_k)(\Psi^3_{t})^2],$$
where \beq \label {eq:coeff-swap-1} \left\{\begin{array}{lcl}
\rho^3(t,T_k)&=&\frac{- 4b^3h^3_ke^{- 2b^3(T_k-t)}}{4(\sigma^3)^2h^3_ke^{-2b^3(T_k-t)}-1}\quad\mbox{with}
\quad h^3_k=\frac{\bar
C^{33}_k}{4 (\sigma^3)^2\bar C^{33}_k-4b^3}\\
 \Gamma^3(t,T_k)&=&-(\sg^3)^2\displaystyle\int_t^{T_k}\rho^3(u,T_k)du.
\end{array}\right.
\eeq

With the use of the explicit expressions for the expectations in (\ref{manip5}), and taking also into account the expression
for $p(t,T)$ in (\ref{pt}), it follows immediately that the
arbitrage-free swap price in (\ref{swa}) can be expressed according
to the following
\begin{prop}\label{Psw}
The price of a payer interest rate swap at $t\leq T_0$ is given by
\begin{equation}
\label{eq:swap-price-2}
\begin{array}{l} P^{Sw}(t; T_0, T_n, R)= \gm\displaystyle\sum_{k=1}^{n} p(t, T_k)
E^{T_k}
\left\{L(T_{k-1}; T_{k-1}, T_k) - R | \mathcal{F}_{t} \right\} \\
\>\>=\displaystyle\sum_{k=1}^{n}p(t,T_{k})\Bigl(D_{t,k}e^{-\rho^1(t,T_k)\Psi^1_{t}-
\rho^2(t,T_k)(\Psi^2_{t})^2-\rho^3(t,T_k)(\Psi^3_{t})^2}-(R\gm+1)\Bigr)\\
\>\>=\displaystyle\sum_{k=1}^{n}\Bigl(D_{t,k}e^{- A_{t,
k}}e^{-\tilde B^1_{t, k}\Psi^1_{t}-\tilde C^{22}_{t,
k}(\Psi^2_{t})^2-\tilde C^{33}_{t,
k}(\Psi^3_{t})^2}-(R\gm+1)e^{-A_{t, k}}e^{- B^1_{t,
k}\Psi^1_{t}-C^{22}_{t, k}(\Psi^2_{t})^2}\Bigr),\end{array}
\end{equation}
where \beq \label {eq:coeff-swap-2}
\begin{array}{l} A_{t, k}:=A(t,T_k),\> B^1_{t, k}:=B^1(t,T_k), \> C^{22}_{t, k}:=C^{22}(t,T_k)\\
\tilde B^1_{t,k}:=B^1_{t, k}+\rho^1(t,T_k), \>\tilde C^{22}_{t,k}:=
C^{22}_{t, k} +\rho^2(t,T_k), \>\tilde C^{33}_{t,k}:=
\rho^3(t,T_k)\\
D_{t,k}:=e^{
A_k}\mathrm{exp}[\Gamma^1(t,T_k)+\Gamma^2(t,T_k)+\Gamma^3(t,T_k)],\end{array}\eeq
with $\rho^i(t,T_k),\,\Gamma^i(t,T_k) \>(i=1,2,3)$ determined
according to (\ref{T12}), (\ref{L8}) and (\ref{eq:coeff-swap-1})
respectively and with $A_k$ as in (\ref{shorth}). \end{prop}

\section{Nonlinear/optional interest rate derivatives}\label{S.4}

In this section we consider the main nonlinear interest rate
derivatives with the Libor rate as underlying. They are also called
{\sl optional derivatives} since they have the form of an option. In
subsection \ref{S.3.3} we shall consider the case of caps and,
symmetrically, that of floors. In the subsequent subsection
\ref{S.3.4} we shall then concentrate on swaptions as options on a
payer swap of the type discussed in subsection \ref{S.3.5}.

\subsection{Caps and floors}\label{S.3.3}

Since floors can be treated in a completely symmetric way to the
caps simply by interchanging the roles of the fixed rate and the
Libor rate, we shall concentrate here on caps. Furthermore, to keep the
presentation simple, we consider here just a single caplet for the
time interval $[T,T+\Delta]$ and for a fixed rate $R$ (recall also
that we consider just one tenor $\Delta$). The payoff of the caplet
at time $T +\Delta$ is thus $\Delta (L(T; T, T +\Delta) - R)^+$,
assuming the notional $N=1$, and its time-$t$ price $P^{Cpl}(t;
T+\Delta, R)$ is given by the following risk-neutral pricing formula
under the forward measure $Q^{T+\Delta}$
$$
P^{Cpl}(t; T+\Delta, R) = \Delta\, p(t,
T+\Delta) E^{{T+\Delta}}
\left\{ \left(L(T; T, T+\Delta) - R \right)^+ \mid \mathcal{F}_t \right\}.
$$
In view of deriving pricing formulas, recall from subsection
\ref{S.3.1} that, under the $(T+\Delta)-$ forward measure, at time
$T$ the factors $\Psi_T^i$ have independent Gaussian distributions (see (\ref{11})) 
with  mean and variance given, for $i=1,2,3$, by
$$E^{T+\Delta}\{\Psi_t^i\}=\bar\al_t^i=\bar\al_t^i(b^i,\sg^i),\quad \quad
Var^{T+\Delta}\{\Psi_t^i\}=\bar\beta_t^i=\bar\beta_t^i(b^i,\sg^i).$$
In the formulas below we shall consider the joint probability
density function of $(\Psi_T^1, \Psi_T^2, \Psi_T^3)$ under the
$T+\Delta$ forward measure, namely, using the independence of the
processes $\Psi_t^i,\,(i=1,2,3)$,
\beq\label{distrib}f_{(\Psi_T^1,\Psi_T^2,\Psi_T^3)}(x_1,x_2,x_3)=\prod_{i=1}^3f_{\Psi_T^i}(x_i)=\prod_{i=1}^3{\cal
N}(x_i,\bar\al^i_T,\bar\beta^i_T), \eeq and use the shorthand notation
$f_i(\cdot)$ for $f_{\Psi_T^i}(\cdot)$ in the sequel. We shall also
write $\bar A, B^1, C^{22}$, $\bar C^{33}$ for the corresponding
functions evaluated at $(T,T+\Delta)$ and given in (\ref{A}),
(\ref{B^1}) and (\ref{eqdiff}) respectively.

Setting $\tilde R:=1+\Delta\,R$, and recalling the first equality
  in (\ref{barp4}), the time-$0$ price  of the
caplet can be expressed as
\beq\label{C1}\begin{array}{l}P^{Cpl}(0;T+\Delta, R)=\Delta\, p(0,
T+\Delta) E^{{T+\Delta}}
\left\{ \left(L(T; T, T+\Delta) - R \right)^+ \right\}\\ \\
\quad=p(0, T+\Delta) E^{{T+\Delta}}  \left\{ \left(
\frac{1}{\bar{p}(T, T+\Delta)} - \tilde{R} \right)^+\right\}\\
\\ \quad=p(0, T+\Delta) E^{{T+\Delta}}  \left\{ \left(e^{\bar
A+(\kappa+1)B^1\Psi_T^1+C^{22}(\Psi_T^2)^2+\bar
C^{33}(\Psi_T^3)^2}-\tilde{R} \right)^+\right\}\\ \\ \quad=p(0,
T+\Delta)\displaystyle\int_{\R^3}\left(e^{\bar
A+(\kappa+1)B^1x+C^{22}y^2+\bar C^{33}z^2}-\tilde{R} \right)^+\\
\hspace{5cm}\cdot
f_{(\Psi_T^1,\Psi_T^2,\Psi_T^3)}(x,y,z)d(x,y,z).\end{array}\eeq

%Before presenting the main result of this subsection, we need to introduce some more notation.
To proceed, we extend to the multi-curve context an idea suggested in \citet{Jamshidian1989} (where it is applied
to the pricing of coupon bonds) by considering the function
\beq\label{C2} g(x,y,z):=\exp[\bar A+(\kappa+1)B^1x+C^{22}y^2+\bar
C^{33}z^2].\eeq Noticing that  $\bar C^{33}(T,T+\Delta)>0$ (see
(\ref{eqdiff}) together with the fact that  $h^3 >0$ and $ 2b^3 +
h^3>0$), for fixed $x,y$ the function $g(x,y,z)$ can be seen to be
continuous and increasing for $z\ge 0$ and decreasing for $z<0$ with
$\lim_{z\to\pm\infty}g(x,y,z)=+\infty.$ It will now be convenient to
introduce some objects according to the following
\begin{df}\label{DB1}
Let a set $M\subset \R^2$ be given by \beq\label{M}
M:=\{(x,y)\in\R^2\mid\>g(x,y,0)\le \tilde R\}\eeq and let $M^c$ be
its complement. Furthermore, for $(x,y)\in M$ let
$$\bar z^1=\bar z^1(x,y)\,,\quad \bar z^2=\bar z^2(x,y)$$ be the solutions of
$g(x,y,z)=\tilde R$. They satisfy $\bar z^1\leq 0 \leq \bar z^2$.
\end{df}

Notice that, for $z\leq \bar z^1\leq 0$ and $z\geq \bar z^2\geq 0$,
we have $g(x,y,z)\geq g(x,y,\bar z^k)=\tilde R$, and for $z \in
(\bar z^1, \bar z^2)$, we have $g(x,y,z)<\tilde R$ . In $M^c$ we
have $g(x,y,z)\ge g(x,y,0)>\tilde R$ and thus no solution of the
equation $g(x,y,z)=\tilde R$.

In view of the main result of this subsection, given in Proposition
\ref{P.3.2} below, we prove as a preliminary the following
\begin{lem}\label{Positivity}
Assuming that the (non-negative) coefficients $b^3,\sg^3$ in the
dynamics (\ref{7}) of the factor $\Psi^3_t$ satisfy the condition
\beq\label{positivity} b^3\ge \frac{\sg^3}{\sqrt{2}},\eeq we have
that $1-2\bar\beta^3_T\,\bar C^{33}>0$, where $\bar
C^{33}=\bar C^{33}(T,T+\Delta)$ 
%with $\bar C^{33}(T,T+\Delta)$ given, for generic $t\le T$, 
is given by (\ref{eqdiff})  and where 
$\bar\beta^3_T=\frac{(\sigma^3)^2}{2b^3}(1-e^{-2b^3T})$ according to 
(\ref{11}).\end{lem}
\begin{proof} From the definitions of $\bar\beta^3_T$ and $\bar
C^{33}$ we may write \beq\label{int1} 1-2\bar\beta^3_T\,\bar C^{33}=
1-\left(1-e^{-2b^3T}\right)\,\frac{2\left(e^{\Delta\,h^3}-1\right)}{2\frac{b^3h^3}{(\sg^3)^2}+\frac{b^3}{(\sg^3)^2}
(2b^3+h^3)\left(e^{\Delta\,h^3}-1\right)}.\eeq Notice next that
$b^3>0$ implies that $1-e^{-2b^3T}\in(0,1)$ and that
$\frac{b^3h^3}{(\sg^3)^2}\ge 0$. From (\ref{int1}) it then follows
that a sufficient condition for $1-2\bar\beta^3_T\,\bar C^{33}>0$ to
hold is that \beq\label{int2} 2\le
\frac{b^3}{(\sg^3)^2}\,(2b^3+h^3).\eeq Given that, see definition
after (\ref{eqdiff}), $h^3=2\sqrt{(b^3)^2+2(\sg^3)^2}\ge 2b^3$, the
condition (\ref{int2}) is satisfied under our assumption
(\ref{positivity}).
\qed
\end{proof}

\begin{prop}\label{P.3.2}
Under assumption (\ref{positivity}) we have that the time-$0$ price of the caplet for the time interval
$[T,T+\Delta]$ and with fixed rate $R$ is given by \beq\label{C3}
\begin{array}{l}P^{Cpl}(0;T+\Delta, R)=p(0,T+\Delta)\Biggl[\displaystyle\int_M e^{\bar
A+(\kappa+1)B^1x+C^{22}(y)^2}\\ \\
\quad\cdot\Bigl[\gm(\bar\al^3_T,\bar\beta^3_T,\bar
C^{33})\left(\Phi(d^1(x,y))+\Phi(-d^2(x,y))\right)\\
\hspace{.4cm}-e^{\bar C^{33}(\bar z^1(x,y))^2}\Phi(d^3(x,y))+e^{\bar
C^{33}(\bar z^2(x,y))^2}\Phi(-d^4(x,y))\Bigr]f_1(x)f_2(y)dxdy\\
\\ \quad+\gm(\bar\al^3_T,\bar\beta^3_T,\bar C^{33})\displaystyle\int_{M^c}e^{\bar
A+(\kappa+1)B^1x+C^{22}(y)^2}f_1(x)f_2(y)dxdy\\ \hspace{5cm}-\tilde
R\,Q^{T+\Delta}\left\{(\Psi_T^1,\Psi_T^2)\in
M^c\right\}\Biggr],\end{array}\eeq where $\Phi(\cdot)$ is the
cumulative standard Gaussian distribution function, $M$ and $M^c$
are as in Definition \ref{DB1}, \beq\label{C4}\begin{cases}
d^1(x,y):=\frac{\sqrt{1-2\bar\beta^3_T\bar C^{33}}\bar z^1(x,y)-(\bar\alpha^3_T-\theta\bar\beta^3_T)}{\sqrt{\bar\beta^3_T}}\\
d^2(x,y):=\frac{\sqrt{1-2\bar\beta^3_T\bar C^{33}}\bar z^2(x,y)-(\bar\alpha^3_T-\theta\bar\beta^3_T)}{\sqrt{\bar\beta^3_T}}\\
d^3(x,y):=\frac{\bar z^1(x,y)-\bar\alpha^3_T}{\sqrt{\bar\beta^3_T}}\\
d^4(x,y):=\frac{\bar z^2(x,y)-\bar\alpha^3_T}{\sqrt{\bar\beta^3_T}}
\end{cases}\eeq
with $\theta:=\frac{\bar\alpha^3_T\left(1-1/\sqrt{1-2\bar\beta^3_T\bar
C^{33}}\right)}{\bar\beta^3_T}$, which by Lemma \ref{Positivity} is
well defined under the given assumption (\ref{positivity}), and with $\gm(\bar\al^3_T,\bar\beta^3_T,\bar
C^{33}):=\frac{e^{(\frac{1}{2}(\theta)^2\bar\beta^3_T-\bar\alpha^3_T\theta)}}
{\sqrt{1-2\bar\beta^3_T\bar C^{33}}}.$
\end{prop}
\begin{rem}\label{RM}
Notice that, once the set $M$ and its complement $M^c$ from
Definition \ref{DB1} are made explicit, the integrals, as well as the
probability in (\ref{C3}), can be computed explicitly.  %analytic
\end{rem}

\begin{proof} On the basis of the sets $M$ and $M^c$ we can continue (\ref{C1}) as
\beq\label{C5}\begin{array}{l}P^{Cpl}(0;T+\Delta, R)=p(0,
T+\Delta)\displaystyle\int_{\R^3}\left(e^{\bar
A+(\kappa+1)B^1x+C^{22}y^2+\bar C^{33}z^2}-\tilde{R} \right)^+ \\
\hspace{5cm}\cdot f_{(\Psi_T^1,\Psi_T^2,\Psi_T^3)}(x,y,z)d(x,y,z) \\ \\
=p(0, T+\Delta)\displaystyle\int_{M\times\R}\left(e^{\bar
A+(\kappa+1)B^1x+C^{22}y^2+\bar C^{33}z^2}-\tilde{R} \right)^+\\
\hspace{5cm}\cdot f_{(\Psi_T^1,\Psi_T^2,\Psi_T^3)}(x,y,z)d(x,y,z)\\ \\
\hspace{1cm}+p(0,
T+\Delta)\displaystyle\int_{M^c\times\R}\left(e^{\bar
A+(\kappa+1)B^1x+C^{22}y^2+\bar C^{33}z^2}-\tilde{R} \right)^+\\
\hspace{5cm}\cdot f_{(\Psi_T^1,\Psi_T^2,\Psi_T^3)}(x,y,z)d(x,y,z)\\ \\
=: P^{1}(0;T+\Delta)+P^{2}(0;T+\Delta).\end{array}\eeq We shall next
compute separately the two terms in the last equality in (\ref{C5})
distinguishing between two cases according to whether $(x,y)\in M$
or $(x,y)\in M^c$.\smallskip

\noindent{\sl Case i):} For $(x,y)\in M$ we have from Definition
\ref{DB1} that there exist $\bar z^1(x,y)\leq 0$ and $\bar
z^2(x,y)\geq 0$ so that for $z\in [\bar z^1, \bar z^2]$ we have
$g(x,y,z)\le g(x,y,\bar z^k)=\tilde R$. For $P^{1}(0;T+\Delta)$ we
now obtain \beq\label{C6}\begin{array}{lcl}P^{1}(0;T+\Delta)&=&p(0,
T+\Delta)\\
&\cdot&\displaystyle\int_{M}e^{\bar A+(\kappa+1)B^1x+C^{22}y^2}
\Biggl( \int_{-\infty}^{\bar z^1(x,y)}(e^{\bar C^{33}z^2}-e^{\bar
C^{33}(\bar z^1)^2})
f_3(z)dz\\
%\hspace{3cm}\cdot\Bigl(\displaystyle  \\
&+&\displaystyle\int^{+\infty}_{\bar z^2(x,y)}(e^{\bar
C^{33}z^2}-e^{\bar C^{33}(\bar
z^2)^2})f_3(z)dz\Biggr)f_2(y)f_1(x)dydx.\end{array}\eeq Next, using
the results of subsection \ref{S.3.1} concerning the Gaussian
distribution $f_3(\cdot)=f_{\Psi_T^3}(\cdot)$, we obtain the
calculations in (\ref{int11}) below, where, recalling Lemma
\ref{Positivity}, we make successively the following changes of
variables: $\zeta:=\sqrt{1-2\bar\beta^3_T\bar C^{33}}z$,
$\theta:=\frac{\bar\alpha^3_T(1-1/\sqrt{1-2\bar\beta^3_T\bar
C^{33}})}{\bar\beta^3_T}$,
$s:=\frac{\zeta-(\bar\alpha^3_T-\theta\bar\beta^3_T)}{\sqrt{\bar\beta^3_T}}$
and where $d^i(x,y),\>i=1,\cdots,4$ are as defined in (\ref{C4})

\begin{align}
\label{int11}
\notag &\displaystyle{\int_{-\infty}^{\bar z^1(x,y)}e^{\bar
C^{33}z^2}f_3(z)dz=
\int_{-\infty}^{\bar z^1(x,y)}e^{\bar C^{33}z^2}\frac{1}{\sqrt{2\pi\bar\beta^3_T}}e^{-\frac{1}{2}\frac{(z-\bar\alpha^3_T)^2}{\bar\beta^3_T}}dz}\\
\notag &\displaystyle{=\int_{-\infty}^{\bar
z^1(x,y)}\frac{1}{\sqrt{2\pi\bar\beta^3_T}}e^{-\frac{1}{2}
\frac{(\sqrt{1-2\bar\beta^3_T\bar
C^{33}}z-\bar\alpha^3_T)^2}{\bar\beta^3_T}}
e^{-\frac{\bar\alpha^3_T(\sqrt{1-2\bar\beta^3_T\bar C^{33}}-1)}{\bar\beta^3_T}z}dz}\\
&\displaystyle{=\int_{-\infty}^{\sqrt{1-2\bar\beta^3_T\bar C^{33}}
\bar
z^1(x,y)}\frac{1}{\sqrt{2\pi\bar\beta^3_T}}e^{-\frac{1}{2}\frac{(\zeta-\bar\alpha^3_T)^2}{\bar\beta^3_T}}
e^{-\frac{\bar\alpha^3_T(1-1/\sqrt{1-2\bar\beta^3_T\bar C^{33}})}{\bar\beta^3_T}\zeta}\frac{1}{\sqrt{1-2\bar\beta^3_T\bar C^{33}}}d\zeta}\\
\notag &\displaystyle{=\frac{1}{\sqrt{1-2\bar\beta^3_T\bar
C^{33}}}\int_{-\infty}^{\sqrt{1-2\bar\beta^3_T\bar C^{33}}\bar
z^1(x,y)}
\frac{1}{\sqrt{2\pi\bar\beta^3_T}}e^{-\frac{1}{2}\frac{(\zeta-\bar\alpha^3_T)^2}{\bar\beta^3_T}}e^{-\theta\zeta}d\zeta}\\
\notag &\displaystyle{=\frac{e^{(\frac{1}{2}
(\theta)^2\bar\beta^3_T-\bar\alpha^3_T\theta)}}{\sqrt{1-2\bar\beta^3_T\bar
C^{33}}}\int_{-\infty}^{d^1(x,y)}\frac{1}{\sqrt{2\pi}}e^{-\frac{s^2}{2}}ds}
\displaystyle{=\frac{e^{(\frac{1}{2}(\theta)^2\bar\beta^3_T-\bar\alpha^3_T\theta)}}{\sqrt{1-2\bar\beta^3_T\bar
C^{33}}}\Phi(d^1(x,y))}.
\end{align}

On the other hand, always using the results of subsection
\ref{S.3.1} concerning the Gaussian distribution
$f_3(\cdot)=f_{\Psi_T^3}(\cdot)$ and making this time the change of
variables $\zeta:=\frac{(z-\bar\alpha^3_T)}{\sqrt{\bar\beta^3_T}}$,
we obtain
\begin{equation}\label{int111}
\begin{split}
&\int_{-\infty}^{\bar z^1(x,y)}e^{\bar C^{33}(\bar z^1)^2}f_3(z)dz=
e^{\bar C^{33}(\bar z^1)^2}\int_{-\infty}^{\bar z^1(x,y)}\frac{1}{\sqrt{2\pi\bar\beta^3_T}}e^{-\frac{1}{2}\frac{(z-\bar\alpha^3_T)^2}{\bar\beta^3_T}}dz\\
&=e^{\bar C^{33}(\bar
z^1)^2}\int_{-\infty}^{d^3(x,y)}\frac{1}{\sqrt{2\pi}}e^{-\frac{1}{2}\zeta^2}d\zeta=e^{\bar
C^{33}(\bar z^1)^2}\Phi(d^3(x,y)).
\end{split}
\end{equation}
Similarly, we have
\begin{equation}\label{int112}
\begin{split}
&\int^{+\infty}_{\bar z^2(x,y)}e^{\bar
C^{33}z^2}f_3(z)dz=\frac{1}{\sqrt{1-2\bar\beta^3_T\bar C^{33}}}
e^{(\frac{1}{2}(\theta)^2\bar\beta^3_T-\bar\alpha^3_T\theta)}\Phi(-d^2(x,y))\\\\
&\int^{+\infty}_{\bar z^2(x,y)}e^{\bar C^{33}(\bar
z^1)^2}f_3(z)dz=e^{\bar C^{33}(\bar z^2)^2}\Phi(-d^4(x,y)).
\end{split}
\end{equation}

\noindent{\sl Case ii):} We come next to the case $(x,y)\in M^{c}$,
for which $g(x,y,z)\geq g(x,y,0)>\tilde R$. For $P^{2}(0;T+\Delta)$
we obtain
\begin{equation}\label{int113}
\begin{split}
&P^{2}(0;T+\Delta)=p(0,T+\Delta)\int_{M^c\times\mathbb{R}}\Bigl(e^{\bar
A+(\kappa+1)B^1x+C^{22}y^2
+\bar C^{33}z^2}-\tilde R\Bigr)\\
&\hspace{5cm}\cdot f_3(z)f_2(y)f_1(x)dzdydx\\
&=p(0,T+\Delta)\Bigl(e^{\bar A}\int_{M^c}e^{(\kappa+1)B^1x+C^{22}y^2}f_1(x)f_2(y)dxdy\int_{\mathbb R}e^{\bar C^{33}z^2}f_3(z)dz\\
&\hspace{4cm}-\tilde R Q^{T+\Delta}[(\Psi^1_{T},\Psi^2_T)\in M^c]\Bigr)\\
&=p(0,T+\Delta)\Bigl(e^{\bar A}\Bigl[\int_{M^c}e^{(\kappa+1)B^1x+C^{22}y^2}f_1(x)f_2(y)dxdy\Bigr]
\frac{e^{(\frac{1}{2}(\theta^3)^2\bar\beta^3_T-\bar\alpha^3_T\theta^3)}}{\sqrt{1-2\bar\beta^3_T\bar C^{33}}}\\
&\hspace{4cm}-\tilde R Q^{T+\Delta}[(\Psi^1_{T},\Psi^2_T)\in
M^c]\Bigr),
\end{split}
\end{equation}
where we computed the integral over $\mathbb{R}$ analogously to
(\ref{int11}).

Adding the two expressions derived for Cases i) and ii), we
obtain the statement of the proposition.\qed
\end{proof}

\subsection{Swaptions}\label{S.3.4}

We start by recalling some of the most relevant aspects of a (payer)
swaption. Considering a swap (see subsection \ref{S.3.5}) for a
given collection of dates $0 \leq T_0 < T_1 < \cdots < T_n$, a
swaption is an option to enter the swap  at a pre-specified
initiation date $T\leq T_0$, which is thus also the maturity of the
swaption and that, for simplicity of notation, we assume to coincide
with $T_0$, i.e. $T=T_0$. The arbitrage-free swaption price at $t\le
T_0$ can be computed as \beq\label{swapt} P^{Swn}(t; T_0, T_n, R) =
p(t, T_0) E^{T_0} \left\{ \left( P^{Sw}(T_0; T_n, R) \right)^{+} |
\mathcal{F}_t \right\},\eeq where we  have used  the shorthand
notation $P^{Sw}(T_0; T_n, R)  = P^{Sw}(T_0; T_0, T_n, R). $

We first state the next Lemma, that follows immediately from the
expression for $\rho^3(t,T_k)$ and the corresponding expression for
$h^3_k$ in (\ref{eq:coeff-swap-1}).
\begin{lem}\label{L.3.1} We have the equivalence
\beq\label{equiv} \rho^3(t,T_k)>0\Leftrightarrow
h^3_k\in\Bigl(0,\frac{1}{4(\sigma^3)^2e^{-2b^3(T_k-t)}}\Bigr).\eeq
\end{lem}
This lemma prompts us to split the swaption pricing problem into two
cases: \beq\label{Cases}\begin{array}{l}
\mbox{\bf Case 1):}\quad h^3_k<0\>
\mbox{or}\,\>h^3_k>\frac{1}{4(\sigma^3)^2e^{-2b^3(T_k-t)}} \\
 \mbox{\bf Case 1):}\quad
0<h^3_k<\frac{1}{4(\sigma^3)^2e^{-2b^3(T_k-t)}}.
\end{array}\eeq
Note from the definition of  $\rho^3(t,T_k)$ that $h^3_k \neq
\frac{1}{4(\sigma^3)^2e^{-2b^3(T_k-t)}}$ and that $h^3_k =0 $ would
imply $\bar C_k^{33}=0$ which corresponds to a trivial case in which
the factor $\Psi^3$ is not present in the dynamics of the spread
$s$, hence the inequalities in Case 1) and Case 2) above are indeed
strict.

To proceed, we shall introduce some more notation. In particular,
instead of only one function $g(x,y,z)$ as in (\ref{C2}), we shall
consider also a function $h(x,y)$, more precisely, we shall define
here the continuous functions \beq\label{T15}
g(x,y,z):=\sum_{k=1}^{n}D_{0,k}e^{- A_{0, k}}e^{-\tilde B^1_{0,
k}x-\tilde C^{22}_{0, k}y^2-\tilde C^{33}_{0, k}z^2}\eeq
\beq\label{T16} h(x,y) :=\sum_{k=1}^{n}(R\gm+1)e^{- A_{0,
k}}e^{-B^1_{0, k} x-C^{22}_{0, k} y^2},\eeq with the coefficients
given by (\ref{eq:coeff-swap-2}) for $t=T_0$. Note that by a slight abuse of notation we write $D_{0, k}$ for $D_{T_0, k}$ and similarly for other coefficients above, always  meaning $t=T_0$ in (\ref{eq:coeff-swap-2}). We distinguish the two
cases specified in (\ref{Cases}):

For {\sl Case 1)} we have (see (\ref{eq:coeff-swap-2}) and Lemma
\ref{L.3.1}) that $\tilde C^{33}_{0,k}=\rho^3(T_0,T_k)<0$ for
all $k=1,\cdots, n$, and so the function $g(x,y,z)$ in (\ref{T15})
is, for given $(x,y)$, monotonically increasing for $z\ge 0$ and
decreasing for $z<0$ with
$$\lim_{z\to\pm\infty}g(x,y,z)=+\infty.$$

For {\sl Case 2)}  we have instead that $\tilde
C^{33}_{0,k}=\rho^3(T_0,T_k)>0$ for all $k=1,\cdots, n$ and so
the nonnegative function $g(x,y,z)$ in (\ref{T15}) is, for given
$(x,y)$, monotonically decreasing for $z\ge 0$ and increasing for
$z<0$ with
$$\lim_{z\to\pm\infty}g(x,y,z)=0.$$

Analogously to Definition \ref{DB1} we next introduce the following
objects
\begin{df}\label{DB2}
Let a set $\bar M\subset \R^2$ be given by \beq\label{barM} \bar
M:=\{(x,y)\in\R^2\mid\>g(x,y,0)\leq h(x,y)\}.\eeq Since $g(x,y,z)$
and $h(x,y)$ are continuous, $\bar M$ is closed, measurable and
connected. Let $\bar M^c$ be its complement. Furthermore, we define
two functions $\bar z^1(x,y)$ and $\bar z^2(x,y)$ distinguishing
between the two Cases 1) and 2) as specified in (\ref{Cases}).
\begin{description}\item[Case 1)]
If $(x,y)\in \bar M$, we have $g(x,y,0)\leq h(x,y)$ and so there exist
$\bar z^1(x,y)\le 0$ and $\bar z^2(x,y)\ge 0$ for which, for
$i=1,2$, \beq\label{barz1}\begin{array}{lcl}g(x,y,\bar
z^i)&=&\displaystyle\sum_{k=1}^{n}D_{0, k}e^{- A_{0, k}}e^{-\tilde
B^1_{0, k}x- \tilde C^{22}_{0, k}y^2-\tilde C^{33}_{0, k}(\bar
z^i)^2}\\&=&\displaystyle\sum_{k=1}^{n}(R\gm+1)e^{- A_{0, k}}e^{-
B^1_{0, k}x-C^{22}_{0, k} y^2}=h(x,y)\end{array}\eeq and, for
$z\not\in[\bar z^1,\bar z^2],$ one has $g(x,y,z)\ge g(x,y,\bar
z^i)$.

If $(x,y)\in \bar M^c$, we have $g(x,y,0)>h(x,y)$ so that
$g(x,y,z)\ge g(x,y,0)>h(x,y)$ for all $z$ and we have no points
corresponding to $\bar z^1(x,y)$ and $\bar z^2(x,y)$ above.\medskip

\item[Case 2)] If $(x,y)\in \bar M$, we have, as for Case 1),
$g(x,y,0)\leq h(x,y)$ and so there exist $\bar z^1(x,y)\le 0$ and $\bar
z^2(x,y)\ge 0$ for which, for $i=1,2$, (\ref{barz1}) holds. However,
this time it is for $z\in[\bar z^1,\bar z^2]$ that one has
$g(x,y,z)\ge g(x,y,\bar z^i)$.

If $(x,y)\in M^c$, then we are in the same situation as for Case 1).
\end{description}\end{df}

Starting from (\ref{swapt}) combined with (\ref{eq:swap-price-2})
and taking into account the set $\bar M$ according to Definition
\ref{DB2}, we can obtain the following expression for the swaption
price at $t=0$. As for the caps, here too we consider the joint
Gaussian distribution
$f_{(\Psi_{T_0}^1,\Psi_{T_0}^2,\Psi_{T_0}^3)}(x,y,z)$ of the factors
under the $T_0-$forward measure $Q^{T_0}$ and we have
\beq\label{Swn}\begin{array}{l} P^{Swn}(0; T_0, T_n, R) = p(0, T_0)
E^{T_0} \left\{ \left(
P^{Sw}(T_0; T_n, R) \right)^{+} | \mathcal{F}_0 \right\}\\
=p(0,T_0)\displaystyle\int_{\mathbb R^3}\Bigl[\displaystyle\sum_{k=1}^{n}D_{0, k}
e^{-A_{0,k}}\mathrm{exp}(-\tilde B^1_{0, k}x-\tilde C^{22}_{0, k}y^2-\tilde C^{33}_{0, k}z^2)\\
\>\>-\displaystyle\sum_{k=1}^{n}(R\gm+1)e^{- A_{0,
k}}\mathrm{exp}(-B^1_{0, k} x- C^{22}_{0, k}
y^2)\Bigr]^+f_{(\Psi^1_{T_0},\Psi^2_{T_0},
\Psi^3_{T_0})}(x,y,z)dxdydz\\
=p(0,T_0)\displaystyle\int_{\bar M\times\mathbb
R}\Bigl[\displaystyle\sum_{k=1}^{n}D_{0, k}
e^{-A_{0, k}}\mathrm{exp}(-\tilde B^1_{0, k}x-\tilde C^{22}_{0, k}y^2-\tilde C^{33}_{0, k}z^2)\\
\>\>-\displaystyle\sum_{k=1}^{n}(R\gm+1)e^{- A_{0, k}}
\mathrm{exp}(-B^1_{0, k} x-C^{22}_{0, k} y^2)\Bigr]^
+f_{(\Psi^1_{T_0},\Psi^2_{T_0},\Psi^3_{T_0})}(x,y,z)dxdydz\\
\>\>+p(0,T_0)\displaystyle\int_{\bar M^c\times\mathbb
R}\Bigl[\sum_{k=1}^{n}D_{0, k}
e^{-A_{0, k}}\mathrm{exp}(-\tilde B^1_{0, k}x-\tilde C^{22}_{0, k}y^2-\tilde C^{33}_{0, k}z^2)\\
\>\>-\displaystyle\sum_{k=1}^{n}(R\gm+1)e^{- A_{0,
k}}\mathrm{exp}(-B^1_{0, k} x-C^{22}_{0, k} y^2)\Bigr]^
+f_{(\Psi^1_{T_0},\Psi^2_{T_0},\Psi^3_{T_0})}(x,y,z)dxdydz\\
=:P^{1}(0; T_0, T_n, R)+P^{2}(0; T_0, T_n, R).\end{array}\eeq

We can now state and prove the main result of this subsection
consisting in a pricing formula for swaptions for the Gaussian
exponentially quadratic model of this paper.  We have
\begin{prop}\label{P.3.3} Assume that the parameters in the model are such that, if $h_k^3$ belongs to Case (1) in (85)  and $h_k^3 >0$, then  $h_k^3 > \frac{1}{4 (\sigma^3)^2 e^{-2 b^3 T_k}}.$ The arbitrage-free
price at $t=0$ of the swaption with payment dates $T_1 < \cdots <
T_n$ such that $\gm=\gm_{k} := T_{k} - T_{k-1}\>(k=1,\cdots,n)$,
with a given fixed rate $R$ and a notional $N=1$, can be computed as
follows where we distinguish between the Cases 1) and 2) specified
in Definition \ref{DB2}.\medskip

\noindent{\bf Case 1)} We have \beq\label{PSWa}
\begin{array}{l}
P^{Swn}(0; T_0, T_n, R)=p(0,T_0)\Biggl\{\displaystyle\sum_{k=1}^{n}e^{- A_{0, k}}
\Biggl[\displaystyle\int_{\bar M}D_{0,k}\mathrm{exp}(-\tilde B^1_{0,k}x-\tilde C^{22}_{0,k}y^2)\\
\hspace{2cm}\cdot\Biggl(\frac{e^{(\frac{1}{2}(\theta_k)^2\bar\beta^3_{T_0}-\bar\alpha^3_{T_0}\theta_k)}}
{\sqrt{1+2\bar\beta^3_{T_0}\tilde C^{33}_{0,k}}}\Phi(d^1_k(x,y))-e^{-\tilde C^{33}_{0,k}(\bar z^1)^2}\Phi(d^2_k(x,y))\\
+\frac{e^{(\frac{1}{2}(\theta_k)^2\bar\beta^3_{T_0}-\bar\alpha^3_{T_0}\theta_k)}}
{\sqrt{1+2\bar\beta^3_{T_0}\tilde C^{33}_{0,k}}}\Phi(-d^3_k(x,y))-e^{-\tilde C^{33}_{0,k}(\bar z^2)^2}\Phi(-d^4_k(x,y))\Biggr)f_2(y)f_1(x)dydx\\
+\displaystyle\int_{\bar M^c}\Bigl(D_{0,k}e^{-\tilde
B^1_{0,k}x-\tilde C^{22}_{0,k}y^2}
\frac{e^{(\frac{1}{2}(\theta_k)^2\bar\beta^3_{T_0}-\bar\alpha^3_{T_0}\theta_k)}}{\sqrt{1+2\bar\beta^3_{T_0}
\tilde C^{33}_{0,k} }}-(R\gm+1)e^{- B^1_{0, k}x-C^{22}_{0, k}
y^2}\Bigr) f_2(y)f_1(x)dydx\Biggr]\Biggr\}.
\end{array}\eeq
\noindent{\bf Case 2)} We have \beq\label{PSWb}
\begin{array}{l}
P^{Swn}(0; T_0, T_n,
R)=p(0,T_0)\Biggl\{\displaystyle\sum_{k=1}^{n} e^{- A_{0, k}}\\
\Biggl[\displaystyle\int_{\bar M}D_{0,k}\mathrm{exp}(-\tilde
B^1_{0,k}x-\tilde C^{22}_{0,k}y^2)
\Bigl(\frac{e^{(\frac{1}{2}(\theta_k)^2\bar\beta^3_{T_0}-\bar\alpha^3_{T_0}\theta_k)}}{\sqrt{1+2\bar\beta^3_{T_0}\tilde
C^{33}_{0,k}}}
\Bigl[\Phi(d^3_k(x,y))-\Phi(d^1_k(x,y))\Bigr]\\
\hspace{2cm}-e^{-\tilde C^{33}_{0,k}(\bar
z^1)^2}\Bigl[\Phi(d^4_k(x,y))-\Phi(d^2_k(x,y))\Bigr]\Bigr)f_2(y)f_1(x)dydx\\
+\displaystyle\int_{\bar M^c}\Bigl(D_{0,k}e^{-\tilde
B^1_{0,k}x-\tilde C^{22}_{0,k}y^2}
\frac{e^{(\frac{1}{2}(\theta_k)^2\bar\beta^3_{T_0}-\bar\alpha^3_{T_0}\theta_k)}}{\sqrt{1+2\bar\beta^3_{T_0}
\tilde C^{33}_{0,k}}}-(R\gm+1)e^{- B^1_{0, k}x-C^{22}_{0, k}
y^2}\Bigr) f_2(y)f_1(x)dydx\Biggr]\Biggr\}.
\end{array}\eeq

The coefficients in these formulas are as specified in
(\ref{eq:coeff-swap-2}) for $t=T_0$, $f_1(x), f_2(x)$ are the Gaussian
densities corresponding to (\ref{distrib}) for $T=T_0$  and the
functions $d^i_k(x,y)$, for $i=1,\ldots,4$ and $k=1, \ldots, n$,
are given by % (compare with (\ref{C4}),
\beq\label{PSWc} \left\{\begin{array}{lcl}
d^1_k(x,y)&:=&\frac{\sqrt{1+2\bar\beta^3_{T_0}\tilde
C^{33}_{0,k}}\bar
z^1(x,y)-(\bar\alpha^3_{T_0}-\theta_k\bar\beta^3_{T_0})}{\sqrt{\bar\beta^3_{T_0}}}\\
d^2_k(x,y)&:=&\frac{\bar
z^1(x,y)-\bar\alpha^3_{T_0}}{\sqrt{\bar\beta^3_{T_0}}}\\
d^3_k(x,y)&:=&\frac{\sqrt{1+2\bar\beta^3_{T_0}\tilde
C^{33}_{0,k}}\bar z^2(x,y)-(\bar\alpha^3_{T_0}
-\theta_k\bar\beta^3_{T_0})}{\sqrt{\bar\beta^3_{T_0}}}\\
d^4_k(x,y)&:=&\frac{\bar
z^2(x,y)-\bar{\alpha}_{T_0}^3}{\sqrt{\bar\beta^3_{T_0}}}\end{array}\right.\eeq
with
$\theta_k:=\frac{\bar{\alpha}^3_{T_0}\left(1-1/\sqrt{1+2\bar{\beta}^3_{T_0}\tilde{C}^{33}_{0,k}}\right)}{\bar{\beta}^3_{T_0}}$,
for $k=1, \ldots, n$,
%with $\theta:=\frac{\bar\alpha^3_T(1-1/\sqrt{1-2\bar\beta^3_T \bar C^{33}})}{\bar\beta^3_T}$
and where $\bar z^1=\bar z^1(x,y),\,\bar z^2=\bar z^2(x,y)$ are
solutions in $z$ of the equation $g(x,y,z)=h(x,y).$

In addition, the mean and variance values for the Gaussian factors
$(\Psi_{T_0}^1,\Psi_{T_0}^2,\Psi_{T_0}^3)$ are here given by
\begin{equation}\label{PSWd}
\left\{\begin{array}{l}
\bar\alpha_{T_0}^1=e^{-b^1{T_0}}\Psi^1_0-\frac{(\sigma^1)^2}{2(b^1)^2}e^{-b^1{T_0}}(1-e^{2b^1{T_0}})
-\frac{(\sigma^1)^2}{(b^1)^2}(1-e^{b^1{T_0}})\Bigr]\\
\bar\beta_{T_0}^1=e^{-2b^1{T_0}}(e^{2b^1{T_0}}-1)\frac{(\sigma^1)^2}{2(b^1)}\\
\bar\alpha_{T_0}^2=e^{-b^2{T_0}}\Psi^2_0\\
\bar\beta^2_{T_0}=e^{-2b^2{T_0}}\displaystyle\int_0^{T_0}e^{2b^2u+4(\sigma^2)^2\bar C^{22}(u,{T_0})}(\sigma^2)^2du\\
\bar\alpha_{T_0}^3=e^{-b^3{T_0}}\Psi^3_0\\
\bar\beta^3_{T_0}=e^{-2b^3{T_0}}\frac{(\sigma^3)^2}{2b^3}(e^{2b^3{T_0}}-1).
\end{array}\right.
\end{equation}
\end{prop}
\begin{rem}\label{barRM}
A remark analogous to Remark \ref{RM} applies here too concerning the
sets $\bar M$ and $\bar M^c$.\end{rem}

\begin{proof} First of all notice that, when  $h_k^3  <0$ or $h_k^3 > \frac{1}{4 (\sigma^3)^2 e^{-2 b^3 T_k}}$  in Case (1), this implies 
$1+2\bar\beta^3_{T_0}\tilde C^{33}_{0,k} \geq 0$ (in Case (2) we always have  $1 + 2 \tilde\beta^3_{T_0} \tilde C^{33}_{0, k} \geq 0$). Hence, the square-root of
the latter expression in the various formulas of the statement of
the proposition is well-defined. This can be checked, similarly as
in the proof of Lemma \ref{Positivity}, by direct computation taking
into account the definitions of $\bar \beta^3_{T_0} $ in
\eqref{PSWd} and of $\tilde C^{33}_{0,k}$ in \eqref{eq:coeff-swap-2}
and \eqref{eq:coeff-swap-1} for $t=T_0$.

We come now to the statement for the\smallskip

\noindent{\bf Case 1.} We distinguish between whether $(x,y)\in\bar
M$ or $(x,y)\in\bar M^c$ and compute separately the two 
terms in the last equality in (\ref{Swn}).\medskip

\noindent{\bf i)} For $(x,y)\in \bar M$ we have from Definition
\ref{DB2} that there exist $\bar z^1(x,y)\le 0$ and $\bar
z^2(x,y)\ge 0$ so that, for $z\not\in[\bar z^1,\bar z^2],$ one has
$g(x,y,z)\ge g(x,y,\bar z^i)$. Taking into account that, under
$Q^{T_0}$, the random variables $\Psi_{T_0}^1, \Psi_{T_0}^2,
\Psi_{T_0}^3$ are independent, so that we shall write
$f_{(\Psi^1_{T_0},\Psi^2_{T_0},\Psi^3_{T_0})}(x,y,z)=f_1(x)f_2(y)f_3(z)$
(see also (\ref{distrib}) and the line following it), we obtain
\beq\label{PSW1}
\begin{array}{l}
P^{1}(0; T_0, T_n,
R)=p(0,T_0)\Bigl[\displaystyle\sum_{k=1}^{n}D_{0,k}e^{-A_{0, k}}
\displaystyle\int_{M}\mathrm{exp}(-\tilde B^1_{0,k}x-\tilde C^{22}_{0,k}y^2)\\
\>\cdot\Bigl(\displaystyle\int_{-\infty}^{\bar
z^1(x,y)}\mathrm{exp}(-\tilde C^{33}_{0,k}z^2)f_3(z)dz
-\displaystyle\int_{-\infty}^{\bar z^1(x,y)}\mathrm{exp}(-\tilde C^{33}_{0,k}(\bar z^1)^2)f_3(z)dz\\
\displaystyle +\int_{\bar z^2(x,y)}^{+\infty}\mathrm{exp}(-\tilde
C^{33}_{0,k}z^2)f_3(z)dz-\int_{\bar
z^2(x,y)}^{+\infty}\mathrm{exp}(-\tilde C^{33}_{0,k}(\bar
z^2)^2)f_3(z)dz\Bigr)f_2(y)f_1(x)dydx\Bigr].
\end{array}\eeq
By means of calculations that are completely analogous to those in
the proof of Proposition \ref{P.3.2}, we obtain, corresponding to
(\ref{int11}), (\ref{int111}) and (\ref{int112}) respectively and
with the same meaning of the symbols, the following explicit
expressions for the integrals in the last four lines 
of \eqref{PSW1}, namely
\begin{equation}\label{fint}
\int_{-\infty}^{\bar z^1(x,y)}e^{-\tilde
C^{33}_{0,k}z^2}f_3(z)dz=\frac{e^{(\frac{1}{2}(\theta_k)^2\bar\beta^3_{T_0}-\bar\alpha^3_{T_0}\theta_k)}}{\sqrt{1+2\bar\beta^3_{T_0}\tilde
C^{33}_{0,k}}}\Phi(d^1_k(x,y)),
\end{equation}
\begin{equation}\label{fint1}
\int_{-\infty}^{\bar z^1(x,y)}e^{-\tilde C^{33}_{0,k}(\bar
z^1)^2}f_3(z)dz=e^{-\tilde C^{33}_{0,k}(\bar
z^1)^2}\Phi(d^2_k(x,y)),
\end{equation}
and, similarly,
\begin{equation}\label{fint2}
\begin{split}
&\int_{\bar z^2(x,y)}^{+\infty}e^{-\tilde
C^{33}_{0,k}z^2}f_3(z)dz=\frac{e^{(\frac{1}{2}(\theta_k)^2\bar\beta^3_{T_0}-\bar\alpha^3_{T_0}\theta_k)}}{\sqrt{1+2\bar\beta^3_{T_0}\tilde
C^{33}_{0,k}}}
\Phi(-d^3_k(x,y)),\\\\
&\int^{+\infty}_{\bar z^2(x,y)}e^{-\tilde C^{33}_{0,k}(\bar
z^2)^2}f_3(z)dz=e^{-\tilde C^{33}_{0,k}(\bar z^2)^2}\Phi(-d^4_k(x,y)),
\end{split}
\end{equation}
where the $d^i_k(x,y)$, for $i=1,\ldots,4$ and $k=1, \ldots, n$, are
as specified in (\ref{PSWc}).\medskip

\noindent{\bf ii)} If $(x,y)\in \bar M^c$ then, according to
Definition \ref{DB2} we have $g(x,y,z)\ge g(x,y,0)>h(x,y)$ for all
$z$. Noticing that, analogously to (\ref{fint}),
$$\int_{\mathbb R}e^{-\tilde C^{33}_{0,k}\zeta^2}f_3(\zeta)d\zeta=
\frac{e^{(\frac{1}{2}(\theta_k)^2\bar\beta^3_{T_0}-\bar\alpha^3_{T_0}\theta_k)}}{\sqrt{1+2\bar\beta^3_{T_0}\tilde
C^{33}_{0,k}}}$$ we obtain the following expression \beq\label{PSW2}
\begin{array}{l}
P^{2}(0; T_0, T_n, R)=p(0,T_0)\displaystyle\sum_{k=1}^{n}e^{-A_{0,
k}}\Bigl[\int_{\bar M^c\times\mathbb R}
\Bigl(D_{0,k}e^{-\tilde B^1_{0,k}x-\tilde C^{22}_{0,k} y^2-\tilde C^{33}_{0,k}z^2}\\
\hspace{3cm}-(R\gm+1)e^{-B^1_{0, k} x-
C^{22}_{0, k} y^2}\Bigr)f_3(z)f_2(y)f_1(x)dzdydx\Bigr]\\
=p(0,T_0)\displaystyle\sum_{k=1}^{n}e^{- A_{0,
k}}\Bigl[D_{0,k}\Bigl(\int_{M^c}e^{-\tilde B^1_{0,k}x-\tilde
C^{22}_{0,k}y^2}f_2(y)f_1(x)dydx\Bigr)
\frac{e^{(\frac{1}{2}(\theta_k)^2\bar\beta^3_{T_0}-\bar\alpha^3_{T_0}\theta_k)}}{\sqrt{1+2\bar\beta^3_{T_0}\tilde C^{33}_{0,k}}}\\
\hspace{3cm}-(R\gm+1)\Bigl(\displaystyle\int_{\bar M^c}e^{-B^1_{0,
k} x- C^{22}_{0, k} y^2}f_2(y)f_1(x)dydx\Bigr)\Bigr].
\end{array}\eeq
Adding the two expressions in i) and ii) we obtain the statement for
the Case 1.\smallskip

\noindent{\bf Case 2).} Also for this case we distinguish between
whether $(x,y)\in\bar M$ or $(x,y)\in\bar M^c$ and, again, compute
separately the two  terms in the last equality in (\ref{Swn}).\smallskip

\noindent{\bf i)} For $(x,y)\in \bar M$ we have that there exist
$\bar z^1(x,y)\le 0$ and $\bar z^2(x,y)\ge 0$ so that, contrary to
Case 1), one has $g(x,y,z)\ge g(x,y,\bar z^i)$ when $z\in[\bar
z^1,\bar z^2]$. It follows that \beq\label{PSW3}
\begin{array}{l}
P^1(0; T_0, T_n,
R)=p(0,T_0)\Biggl[\displaystyle\sum_{k=1}^{n}D_{0,k}e^{-A_{0, k}}
\displaystyle\int_{\bar M}\mathrm{exp}(-\tilde B^1_{0,k}x-\tilde C^{22}_{0,k}y^2)\\
\quad \cdot\Biggl(\displaystyle\int_{\bar z^1(x,y)}^{\bar
z^2(x,y)}\mathrm{exp}(-\tilde C^{33}_{0,k}z^2)f_3(z)dz
-\displaystyle\int_{\bar z^1(x,y)}^{\bar
z^2(x,y)}\mathrm{exp}(-\tilde C^{33}_{0,k}(\bar
z^1)^2)f_3(z)dz\Biggr) f_2(y)f_1(x)dydx\Biggr]
\\ \hspace{7cm}\\
=p(0,T_0)\Biggl[\displaystyle\sum_{k=1}^{n}D_{0,k}e^{-A_{0, k}}
\displaystyle\int_{\bar M}\mathrm{exp}(-\tilde B^1_{0,k}x-\tilde C^{22}_{0,k}y^2)\\
\hspace{4cm}\cdot\Biggl(
\frac{e^{(\frac{1}{2}(\theta_k)^2\bar\beta^3_{T_0}-\bar\alpha^3_{T_0}\theta_k)}}{\sqrt{1+2\bar\beta^3_{T_0}\tilde
C^{33}_{0,k}}}\left(\Phi(d^3_k(x,y))-\Phi(d^1_k(x,y))\right)\\
\hspace{1cm}-e^{-\tilde C^{33}_{0,k}(\bar
z^1)^2}\left(\Phi(d^4_k(x,y))-\Phi(d^2_k(x,y))\right)\Biggr)f_2(y)f_1(x)dydx\Biggr],
\end{array}\eeq
where we have made use of (\ref{fint}) and (\ref{fint1}),
(\ref{fint2}).
\bigskip

\noindent{\bf ii)} For $(x,y)\in \bar M^c$ we can conclude exactly
as we did it for Case 1) and, by adding the two expressions in i)
and ii), we obtain the statement also for Case 2). \qed
\end{proof}

%\bibliographystyle{chicago}
%\bibliography{references}
\bibliographystyle{plainnat}
\bibliography{ref-cva-ird}

\end{document}